\documentclass{journallet}
\usepackage{graphicx,amssymb,amsmath,amsfonts,fancyhdr,wrapfig,here}
\usepackage{times}
\usepackage{latexsym}
\usepackage{bbm}
\usepackage{enumerate}
\usepackage{url}
\usepackage{subfig}
\usepackage{microtype}%if unwanted, comment out or use option "draft"
\def\GG{{\mathcal G}}
\def\weight{{\rm weight}}
\def\area{{\rm area}}

\title{Universal point sets for planar three-trees\footnote{Fulek gratefully acknowledge support from the Swiss National Science Foundation Grant No.  200021-125287/1 and ESF Eurogiga project GraDR as GA\v{C}R GIG/11/E023. Research by T\'oth was supported, in part, by the NSERC grant RGPIN 35586 and the NSF grant CCF-0830734. Preliminary results have appeared in the Proceedings
of the 12th Algorithms and Data Structures Symposium
(London, ON, 2013), LNCS 8037, Springer, pp. 341–352.}}

\author{Radoslav Fulek\thanks{Columbia University, New York City, NY, USA.
Email: \texttt{radoslav.fulek@gmail.com}}
\and
Csaba D. T\'{o}th\thanks{California State University, Northridge, CA, USA;
University of Calgary, Calgary, AB, Canada;
and Department of Computer Science, Tufts University, Medford, MA, USA.
Email: \texttt{cdtoth@acm.org}}
}

\begin{document}
\maketitle

\begin{abstract}
\noindent
For every $n\in \mathbb{N}$, we present a set $S_n$ of $O(n^{3/2}\log n)$ points in the plane
such that every planar 3-tree with $n$ vertices has a straight-line embedding in the plane
in which the vertices are mapped to a subset of $S_n$. This is the first subquadratic
upper bound on the size of universal point sets for planar 3-trees, as well as for
the class of 2-trees and serial parallel graphs.
 \end{abstract}

\section{Introduction}

Every planar graph has a \emph{straight-line embedding} in the plane~\cite{Far48} where the vertices
are mapped to distinct points and the edges to pairwise noncrossing straight line segments between the corresponding vertices.
A set $S\subset\mathbb{R}^2$ of points in the plane is called \emph{$n$-universal} if every $n$-vertex planar graph has a straight-line embedding in $\mathbb{R}^2$ such that the vertices are mapped into a subset of $S$. Similarly, $S\subset\mathbb{R}^2$ is \emph{$n$-universal for a family $\GG$} of planar graphs if every $n$-vertex planar graph in $\GG$ has a straight-line embedding in $\mathbb{R}^2$ such that the vertices are mapped into a subset of $S$. It is a longstanding open problem to determine the minimum size $f(n)$ of an $n$-universal point set for all $n\in \mathbb{N}$.
Our main result is that there is an $n$-universal point set of size $O(n^{3/2}\log n)$ for the class of planar graphs of treewidth at most three.

\begin{theorem}\label{thm:main}
For every $n\in \mathbb{N}$, there is an $n$-universal point set of size $O(n^{3/2}\log n)$ for planar 3-trees.
\end{theorem}

A graph is called a \emph{$k$-tree}, for some $k\in \mathbb{N}$, if it can be constructed by the following iterative process: start with a $k$-vertex clique and successively add new vertices such that each new vertex has exactly $k$ neighbors that form a clique in the current graph. For example, 1-trees are the same as trees; 2-trees are maximal series-parallel graphs, and include also all outerplanar graphs. In general, $k$-trees are the maximal graphs with treewidth $k$. A planar 3-tree is a 3-tree that is planar. Theorem~\ref{thm:main} is the first subquadratic upper bound on the size of $n$-universal point sets for planar 3-trees, for 2-trees, and for series-parallel graphs.

\paragraph{Related previous work.}
In a pivotal paper, de~Fraysseix, Pach and Pollack~\cite{FPP90} showed that an $n$-universal set must have at least $n+(1-o(1))\sqrt{n}$ points. Chrobak and Karloff~\cite{CK89} improved the lower bound to $1.098n$ and later Kurowski~\cite{Kur04} to $(1.235-o(1))n$.
This is the currently known best lower bound for $n$-universal sets in general.
De~Fraysseix et al.~\cite{FPP90} and Schnyder~\cite{Sch90} independently showed that there are $n$-universal sets of size $O(n^2)$. In fact, an $(n-1)\times (n-1)$ section of the integer lattice is $n$-universal~\cite{CK97,Sch90} for every $n\geq 3$. Alternatively, an $\frac{4}{3}n\times\frac{2}{3}n$ section of the integer lattice is also $n$-universal~\cite{Bra08}. The quadratic upper bound is the best possible if the point set is restricted to sections of the integer lattice: Frati and Patrignani~\cite{FP07} showed (based on earlier work by Dolev~et al.~\cite{DLT84}) that if a rectangular section of the integer lattice is $n$-universal, then it must contain at least $n^2/9+\Omega(n)$ points.

Grid drawings have been studied intensively due to their versatile applications.
It is known that sections of the integer lattice with $o(n^2)$ points are $n$-universal for certain classes of graphs. For example, Di Battista and Frati~\cite{BF09} proved that an $O(n^{1.48})$ size integer grid is $n$-universal for  \emph{outerplanar} graphs. Frati~\cite{Fra10} showed that 2-trees on $n$ vertices require a grid of size at least $\Omega(n2^{\sqrt{\log n}})$. Biedl~\cite{Bie11} observed that the grid embedding of all $n$-vertex 2-trees requires an $\Omega(n)\times \Omega(n)$ section of the integer lattice \emph{if} the combinatorial embedding (i.e., all vertex-edge and edge-face incidences) is given. On the other hand, Zhou et al.~\cite{ZHN12} showed recently that every  $n$-vertex series-parallel graph, and thus, every 2-tree, has a straight-line embedding in a $\frac{2}{3}n\times\frac{2}{3}n$ section of the integer lattice and a section of the integer lattice of area $0.3941n^2$.
Researchers have studied classes of planar graphs that admit $n$-universal point sets of size $o(n^2)$. A classical result in this direction, due to Gritzmann~et al.~\cite{GMPP91} (see also~\cite{Bos02}), is that every set of $n$ points in general position is $n$-universal for \emph{outerplanar graphs}. Angelini~et al.~\cite{ABK+11} generalized this result and showed that there exists an $n$-universal point set of size $O(n(\log n/\log \log n)^2)$ for so-called \emph{simply nested} planar graphs. A planar graph is simply nested if it can be reduced to an outerplanar graph by successively deleting chordless cycles from the boundary of the outer face. Recently, Bannister et al.~\cite{BCDE13+} found $n$-universal point sets of size $O(n \log n)$ for simply nested planar graphs, and $O(n \ {\rm polylog}\ n)$ for planar graphs of bounded pathwidth.

Theorem~\ref{thm:main} provides a new broad class of planar graphs that admit subquadratic $n$-universal sets.

Algorithmic questions pertaining to the straight-line embedding of planar graphs have also been studied.
The \emph{point set embeddability} problem asks whether a given planar graph $G$ has a straight-line embedding such that the vertices are mapped to a given point set $S\subset \mathbb{R}^2$. The problem is known to be NP-hard~\cite{Cab06}, and remains NP-hard even for 3-connected planar graphs~\cite{DM12}, triangulations and 2-connected outerplanar graphs~\cite{BV11}. However, it has a polynomial-time solution for 3-trees~\cite{DMN+11,NMR12,DM13}.
In a \emph{polyline} embedding of a plane graph, the edges are represented by pairwise noncrossing polygonal paths. Biedl~\cite{Bie11} proved that every 2-tree with $n$ vertices has a polyline embedding where the
vertices are mapped to an $O(n)\times O(\sqrt{n})$ section of the integer lattice, and each edge is a polyline with at most two bends. Everett et al.~\cite{ELL+10} showed that there is a set $S_n$ of $n$ points in the plane, for every $n\in \mathbb{N}$, such that every $n$-vertex planar graph has a polyline embedding with at most one bend per edge on $S$. Dujmovi\'{c} et al.~\cite{DEL+11} constructed a point set $S_n'$ of size $O(n^2/\log n)$ for all $n\in \mathbb{N}$ such that every $n$-vertex planar graph has a polyline embedding with at most one bend per edge in which the vertices as well as all bend points of the edges are mapped to $S_n'$.

\paragraph{Organization.}
We briefly review some structural properties of planar 3-trees (Section~\ref{sec:prop3}),
then construct a point set $S_n\subset\mathbb{R}^2$ for every $n\in \mathbb{N}$ (Section~\ref{sec:constr3}),
and show that it is $n$-universal for planar 3-trees (Section~\ref{sec:alg3}).

\section{Basic Properties of Planar Three-Trees\label{sec:prop3}}

A graph $G$ is a \emph{planar 3-tree} if it can be constructed by the following iterative procedure. Initially, let $G=K_3$, the complete graph with three vertices. Successively augment $G$ by adding one new vertex $u$ and three new edges that join $u$ to three vertices of a triangle such that no two vertices are connected to all the vertices of the same triangle. A planar 3-tree can be embedded in the plane simultaneously with the iterative process: the initial triangle forms the outer-face and each new vertex $u$ is inserted in the interior of the face corresponding to the triangle it is attached to.

The iterative augmentation process that produces a 3-tree $G$ can be represented by a rooted tree $T=T(G)$ as follows (this is called a \emph{face-representative tree} in~\cite{HMRS12}). Refer to Fig.~\ref{fig:tree3}. The nodes of $T$ correspond to the triangles of $G$. For convenience we denote a vertex of $T$ by its corresponding triangle in $G$. The root of $T$ corresponds to the initial triangle of $G$. When $G$ is augmented by a new vertex $u$ connected to the vertices of the triangle $\Delta=v_1v_2v_3$, we attach three new leaves to $\Delta$ corresponding to the triangles $v_1v_2u$, $v_1uv_3$ and $uv_2v_3$.
For a node $\Delta$ of $T$, let $T_\Delta$ denote the subtree of $T$ rooted at $\Delta$.
Let $V_\Delta$ denote the set of vertices of $G$ embedded in the interior of $\Delta$.

\begin{figure}[htb]
\centering
\includegraphics[scale=0.75]{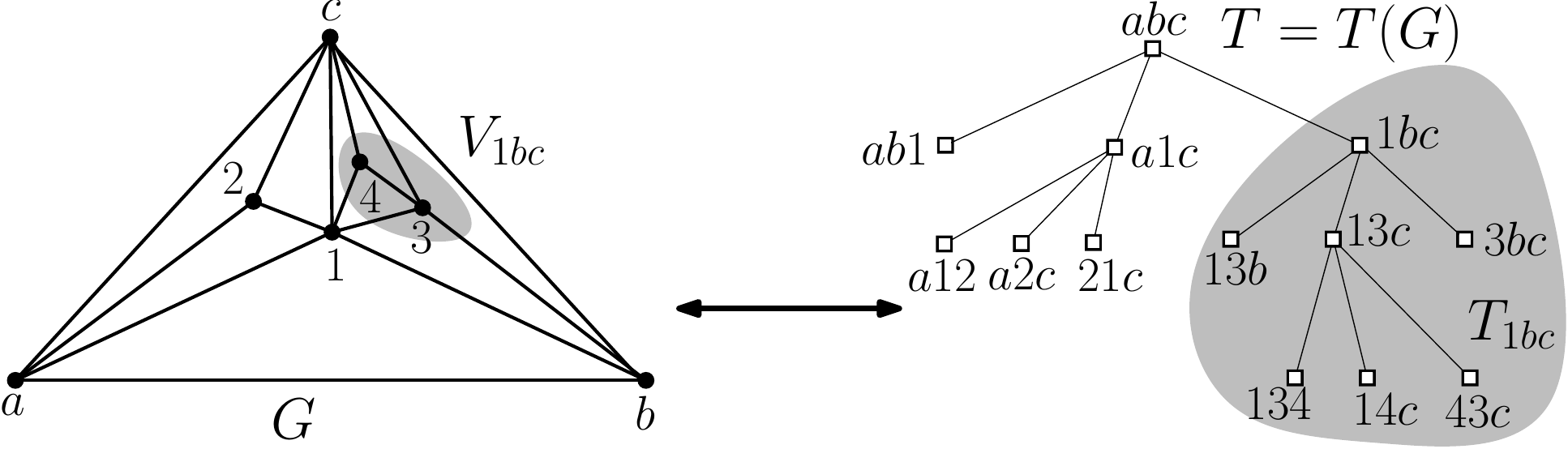}
\caption{Left: a 3-tree, constructed from the initial triangle $abc$ by successively
adding new vertices $1,\ldots , 4$.
Right: The corresponding tree $T=T(G)$.
The gray region indicate the subtree $T_{1bc}$ rooted at $1bc$, and its corresponding vertex set
$V_{1bc}\subseteq V(G)$.}
\label{fig:tree3}
\end{figure}

In Section~\ref{sec:alg3}, we embed the vertices of a planar 3-tree on a point set by traversing the tree $T$ from the root. The initial triangle $abc$ will be the outer face in the embedding such that the edge $ab$ is a horizonal line segment, and the vertex $c$ is the top vertex (i.e., it has  maximal $y$-coordinate). We then successively insert the remaining $n-3$ vertices of $G$, each of which subdivides a triangular face into three triangles. We label the vertices of each triangle of $G$ as \emph{left}, \emph{right} and \emph{top} vertex, respectively. These labels are assigned (without knowing the specifics of our embedding algorithm) as follows. Label the three vertices of the initial triangle in $G$ arbitrarily as \emph{left}, \emph{right} and \emph{top}, respectively. When $G$ is augmented by a new vertex $u$ and edges $uv_1$, $uv_2$, and $uv_3$, where $v_1$ is the left, $v_2$ is the right, and $v_3$ is the top vertex of an existing triangle $v_1v_2v_3$, then let $v_1$, $v_2,$ and $v_3$ keeps their labels left, right, top, respectively) in the new triangles $v_1v_2u$, $v_2v_3u$ and $v_1v_3u$; while vertex $u$ becomes the top vertex of $v_1v_2u$, the left vertex of $v_2v_3u$, and the right vertex of $v_1v_3u$.
The triangles $v_1v_2u$, $v_1uv_3$ and $uv_2v_3$, respectively, will be called the \emph{bottom}, \emph{left} and \emph{right} triangles within $v_1v_2v_3$. In the tree $T=T(G)$, the three children of a node corresponding to a vertex can be labeled as  \emph{bottom}, \emph{left}, and \emph{right} child, analogously.

\paragraph{Weighted Nodes in $T(G)$.}
Let $G$ be a planar 3-tree with $n$ vertices. Our embedding algorithm (in Section~\ref{sec:alg3}) is guided by the tree $T=T(G)$, which represents an incremental process that constructs $G$ from a single triangle. Recall that $T_\Delta$ denotes the subtree of $T$ rooted at a node $\Delta$; and $V_\Delta$ denotes the set of vertices of $G$ that correspond to nodes in $T_\Delta$. Let the \emph{weight} of a node $\Delta$ of $T$ be $\weight(\Delta)=|V_\Delta|$. The tree $T$ is a partition tree: for every node $\Delta$, $\weight(\Delta)$ equals one plus the total weight of the children of $\Delta$.

Let $\alpha\in (0,1]$ be a constant. A node $\Delta$ is \emph{heavy} (resp., \emph{light}) in $T$ if its weight is at least (resp., less than) $n^{\alpha}$. We designate some of the nodes in $T$ as \emph{hubs} recursively in a top-down traversal of the tree $T$: Let the root of $T$ be a hub.
Let a node $\Delta\in V_T$ be a hub if $n^\alpha\leq \weight(\Delta)\leq \weight(\Delta')-n^\alpha$, for every hub $\Delta'$ that is an ancestor of $\Delta$. We note a few immediate consequences of the definition.

\begin{lemma}\label{lem:sibling}
If $\Delta_1,\Delta_2\in V_T$ are heavy siblings, then they are both hubs.
\end{lemma}
\begin{proof}
Let $\Delta$ denote the common parent of $\Delta_1$ and $\Delta_2$. Then $\weight(\Delta)>\weight(\Delta_1)+\weight(\Delta_2)$.
Since both $\Delta_1$ and $\Delta_2$ are heavy, we have $\weight(\Delta)>2n^\alpha$. If $\Delta'=\Delta$ or $\Delta'$ is an
ancestor of $\Delta$, we have $n^\alpha\leq \weight(\Delta_i)\leq \weight(\Delta')-n^\alpha$ for $i=1,2$.
\end{proof}

\begin{lemma}\label{lem:hubs}
The tree $T(G)$ has at most $2n^{1-\alpha}$ hubs.
\end{lemma}
\begin{proof}
Let $T'$ be the subtree of $T$ induced by all heavy nodes.
By definition, every hub of $T$ is in $T'$.
Denote by $T''$ the tree obtained from $T'$ by adding a sibling leaf
to every hub that is a single child in $T'$.
Note that every hub of $T$ is in $T''$, and its parent has at least two children in $T''$.
The tree $T''$ has at most $n^{1-\alpha}$ leaves, since every leaf of $T''$ accounts for at
least $n^{\alpha}$ vertices of $G$. Therefore, $T''$ has at most $n^{1-\alpha}-1$ vertices
with two or three children. Together with the root, there are at most $2n^{1-\alpha}$ hubs in $T$.
\end{proof}

\section{Construction of a Point Set\label{sec:constr3}}

Let $\alpha\in (0,1]$ be a constant. In this section, we construct a point set
$S_n$ of size $\Theta(n^{2-\alpha}\log n)$ for every $n\in \mathbb{N}$. In Section~\ref{sec:alg3},
we show that for $\alpha=1/2$, the point set $S_n$ of size $\Theta(n^{3/2}\log n)$ is $n$-universal for planar 3-trees. Assume for the remainder of this section that $n^{\alpha}=2^q$ for some positive integer $q\in \mathbb{N}$; otherwise let $S_n=S_{n'}$ for $n'=2^{\lceil \log_2 n\rceil}$.

The point set $S_n$ is constructed in two steps: we first choose a ``sparse'' set $B_n$ of points from a $14n\times 14n$ section of the integer lattice, and then ``stretch'' the points by the transformation $(x,y)\rightarrow (x,(28n)^y)$, as described below.

\paragraph{Sparse grid.}
Let $A_n=\{(i,j)\in \mathbb{Z}^2: 0 \le i,j\le 14n\}$ be an $14n\times 14n$ section of the integer lattice.
Let $B_n\subset A_n$ be the set of points in $A_n$ with at least one of the following
three properties (refer to Fig.~\ref{fig:grid}):
\begin{itemize}\itemsep -2pt
\item $(i,j)$ such that $n^\alpha| i j$;
\item $(i+k,j+k)$ such that $n^\alpha|i$, $n^\alpha|j$, and $k\in \{1,2,\ldots  n^\alpha\}$ (\emph{forward diagonals});
\item $(i+k,j-k)$ such that $n^\alpha|i$, $n^\alpha|j$, and $k\in \{1,2,\ldots  n^\alpha\}$ (\emph{backward diagonals}).
\end{itemize}
Note that for every $0\leq i\leq 14n$, if $n^\alpha|i$, then all points $(i,j)\in A_n$ are in $B_n$. We say that these points form a \emph{full row}. Similarly, for every $0\leq j\leq 14n$, if $n^\alpha|j$, then all points $(i,j)\in A_n$ are in $B_n$, forming a \emph{full column}. The points $(i,j)\in A_n$, with $n^\alpha|i$ and $n^\alpha|j$ lie at the intersection points of full rows and full columns.

\begin{figure}[htb]
\centering
\centering
\subfloat[]{
\label{fig:gridB}
\includegraphics[width=2in]{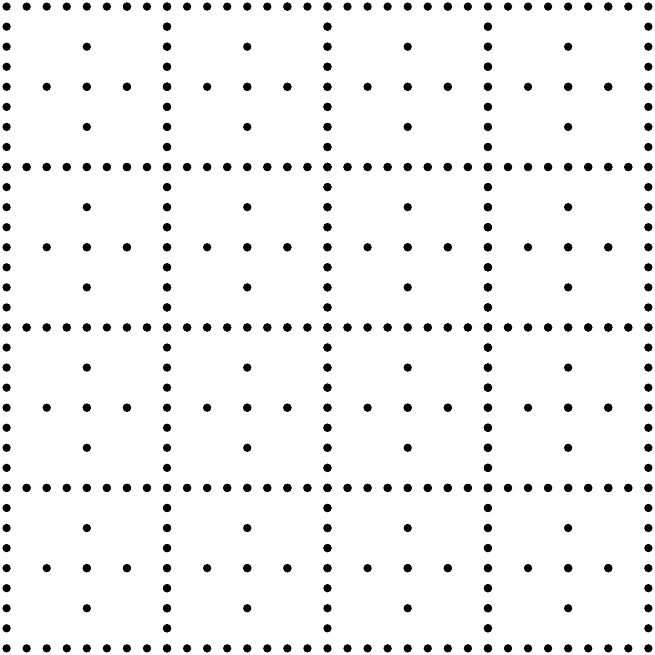}
}
\hspace{2cm}
\subfloat[]{
\label{fig:gridD}
\includegraphics[height=2in]{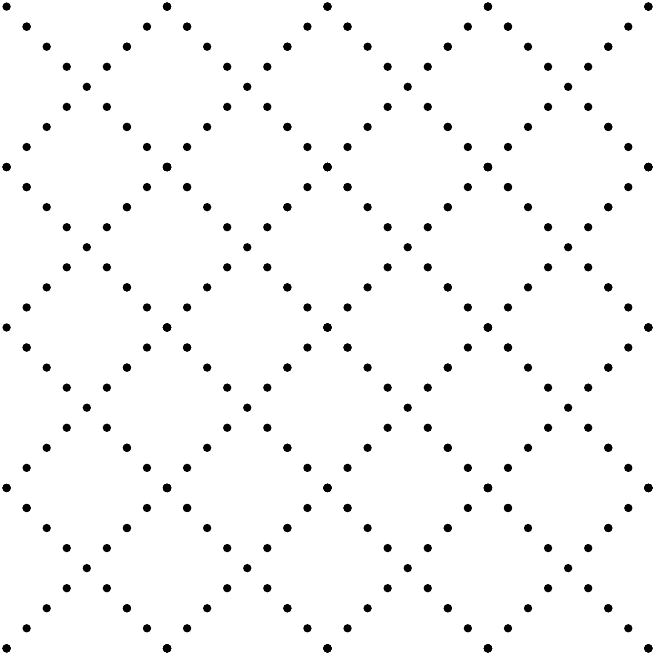}
}
\caption{(a) A patten of points $(i,j)\in \mathbb{Z}^2$ with $8|ij$.
(b) A pattern of forward and backward diagonals.
\label{fig:grid}}
\end{figure}

\paragraph{Stretched grid.}
We deform the plane by the following transformation.
$$\tau: \mathbb{R}^2\rightarrow \mathbb{R}^2, \hspace{2cm} (x,y)\rightarrow (x,(28n)^y).$$
For an integer point $(i,j)\in \mathbb{Z}^2$, we use the shorthand notation
$\tau(i,j)=\tau((i,j))$. If $A\subset \mathbb{R}^2$ is a rectangular section of the integer
lattice (a \emph{grid}), then we call the point set $\tau(A)=\{\tau(p):p\in A\}$ a
\emph{stretched grid}. Note that $\tau$ translates every point vertically,
and it translates points of the same $y$-coordinate by the same vector.

\paragraph{Universal point set for 3-trees.}
We are now in a position to define $S_n$. Let $S_n=\tau(B_n)$.

Similarly to~\cite{BMN11}, our illustrations show the ``unstretched'' point set $B_n=\tau^{-1}(S_n)$ instead of $S_n$. The transformation $\tau^{-1}$ maps line segments between points in $S_n$ to Jordan arcs between grid points in $B_n$. In our figures, line segments are drawn as Jordan arcs that correctly represent the above-below relationship between segments and points (Fig.~\ref{fig:stretching}).

\subsection{Properties of Sparse Grids}

We first show that $B_n$ contains $O(n^{2-\alpha}\log n)$ points.

\begin{lemma}\label{lem:count}
For every $\alpha\in (0,1]$, the sparse grid $B_n$ contains $O(n^{2-\alpha}\log n)$ points.
\end{lemma}
\begin{proof}
It is enough to consider the case that $n^{\alpha}=2^q$ for some positive integer $q\in \mathbb{N}$. We count first the points $(i,j)\in A_n$ such that $2^q|i j$. The grid $A_n$ has $14n+1$ rows and $14n+1$ columns. If $(i,j)\in A$, then $2^q|ij$, hence $2^k|i$ and $2^{q-k}|j$ for some $k=0,1,\ldots, q$. There are exactly $14n/2^k+1$ values $j$, $0\leq j\leq 14n$, with $2^{q-k}|j$, and so the number of pairs $(i,j)\in A_n$ with $2^k|i$ and $2^{q-k}|j$ is $(14n/2^k+1)(14n/2^{q-k}+1)$. Therefore, the total number of points $(i,j)\in A_n$ is bounded above by
$$\sum_{k=0}^q \left(\frac{14n}{2^k}+1\right)\left(\frac{14n}{2^{q-k}}+1\right) \leq
q\cdot \frac{(14n)^2}{2^q} +2\cdot 2\cdot 14n+1 =O(n^{2-\alpha}\log n).$$

Consider now the points of the forward and backward diagonals. Every $n^{\alpha}$-th row and every $n^{\alpha}$-th column is full, and so $B_n$ contains $(14n^{1-\alpha}+1)\cdot (14n^{1-\alpha}+1)=O(n^{2-\alpha})$ points lying at a full column and a full row. Each such point column generates at most $n^\alpha-1$ points in a forward diagonal and $n^\alpha-1$ points in a backward diagonal. The total number of these points is $O(n^{2-\alpha})$.
\end{proof}

The convex hull of $A_n$, denoted ${\rm conv}(A_n)$, is a closed square of side length $14n$.
For an axis-aligned (open) rectangle $R=(x_1,x_2)\times (y_1,y_2)$, we introduce the following parameters:
\begin{itemize}\itemsep -2pt
\item the \emph{width} of $R$ is $w(R)=x_2-x_1$;
\item the \emph{height} of $R$ is $h(R)=y_2-y_1$;
\item the \emph{area} of $R$ is $\area(R)=(x_2-x_1)(y_2-y_1)$.
\end{itemize}
For example, if $R_0={\rm int}({\rm conv} (A_n))$, then $w(R_0)=14n$, $h(R_0)=14n$, and $\area(R_0)=(14n)^2$.

Note also that the sparse grid contains at least one point in every sufficiently large axis-aligned rectangle.

\begin{lemma}\label{lem:dense}
Let $R\subset {\rm conv}(A_n)$ be an open axis-aligned rectangle such that
$w(R)>1$, $h(R)>1$ and $\area(R)\geq 4\cdot n^\alpha$.
Then $B_n\cap R\neq \emptyset$.
\end{lemma}
\begin{proof}
Let $k\in \mathbb{N}$ be the largest integer such that $2^k<w(R)$; and $\ell\in \mathbb{N}$
be the largest integer such that $2^\ell<h(R)$.
Then $2^k\geq w(R)/2$, $2^\ell\geq h(R)/2$, and so $2^{k+\ell}\geq \area(R)/4 \geq n^\alpha=2^q$.
That is, we have $k+\ell\geq q$. Now $R$ intersects a vertical line $\ell_x: x=i$ such that $2^k|i$;
and it also intersect a horizontal line $\ell_y:y=j$ such that $2^\ell|j$.
The point $(i,j)\in A_n$ is in $R$, and $2^{k+\ell}|ij$, as required.
\end{proof}

\begin{lemma}\label{lem:dense2}
Let $R\subset {\rm conv}(A_n)$ be an open axis-aligned rectangle such that
[$w(R)>1$ and $h(R)>n^\alpha$] or [$w(R)>n^\alpha$ and $h(R)>1$].
Then there is a point $(i,j)\in B_n\cap R$ on a forward and a backward diagonal.
\end{lemma}
\begin{proof}
On any vertical line $\ell: x=i$ (resp., horizontal line $\ell: y=i)$, $0\leq i\leq 14n$, the distance between two consecutive points on forward diagonals is $n^\alpha$. If $h(R)\geq n^\alpha$, then $R$ contains a point of a forward diagonal on any line $\ell: x=i$ that intersects $R$. Similarly, if $w(R)\geq n^\alpha$, then $R$ contains a point of a forward diagonal on any line $\ell: y=i$ that intersects $R$.
\end{proof}

In our embedding algorithms (in Section~\ref{sec:alg3}), we translate some points $p\in B_n$ vertically or horizontally by
$n^\alpha$. We note here that the translated image of $p\in B_n$ is either in $B_n$ or outside of the bounding box of $B_n$.

\begin{lemma}\label{lem:shift}
Let $p\in B_n$, and translate $p$ by a horizontal or vertical vector of length $n^\alpha$
to another point $p'$. If $p'\in {\rm conv}(B_n)$, then $p'\in B_n$.
\end{lemma}
\begin{proof}
Assume first that $p=(i,j)\in A_n$ such that $n^\alpha|ij$. It is clear that $n^\alpha|(i\pm n^\alpha)j$
and $n^\alpha| i(j\pm n^\alpha)$. Assume now that $p=(i+k,j+k)\in A_n$ such that $n^\alpha|i$ and $n^\alpha|j$. Then $n^\alpha|i\pm n^\alpha$ and $n^\alpha| j\pm n^\alpha$.
In both cases, if $p'$ is still within ${\rm conv}(B_n)$, then it is also in $B_n$.
\end{proof}

\subsection{Properties of Stretched Grids}

The purpose of transformation $\tau$ is to establish the following
property for the stretched grid $\tau(A_n)$.

\begin{lemma}\label{lem:strechedsegment}
Let $(a_1,b_1), (a_2,b_2), (a_3,b_3)\in A_n$ such that $(a_2,b_2)$ lies in the interior of the axis-aligned rectangle spanned by $(a_1,b_1)$ and $(a_3,b_3)$ (formally, $a_1<a_2<a_3$ and either $b_1<b_2<b_3$ or $b_3<b_2<b_1$). Then $\tau(a_2,b_2)$ lies below the line segment between $\tau(a_1,b_1)$ and $\tau(a_3,b_3)$. (See Fig.~\ref{fig:stretching}.)
\end{lemma}
\begin{proof}
We may assume that $b_1<b_2<b_3$, since the other case can be treated analogously.
Denote by $L$ the line through $\tau(a_1,b_1)$ and $\tau(a_3,b_3)$.
Consider the following function in two variables:
$$D(x,y)=\left|
\begin{matrix}
   1 & 1 &1 \\
   a_1& a_3 & x \\
   c_1& c_3 & y
\end{matrix} \right|
=a_3y-c_3x-a_1y+c_1x+a_1c_3-c_1a_3.$$
The function $D(x,y)$ is negative for all the points below
$L$ and positive for all the points above $L$.
For $\tau(a_2,b_2)$, we have
$D(a_2,c_2)= c_3(a_1-a_2)+a_2c_1+c_2(a_3-a_1)-c_1a_3< c_3(a_1-a_2)+a_2c_1+c_2a_3<-c_3+14nc_1+14nc_2< -(28n)^{b_3}+14n(28n)^{b_3-1}+14n(28n)^{b_3-1}=0$.
Hence $\tau(a_2,b_2)$ is below the line $L$, as required.
\end{proof}

\begin{figure}[htb]
\centering
\includegraphics[width=0.3\textwidth]{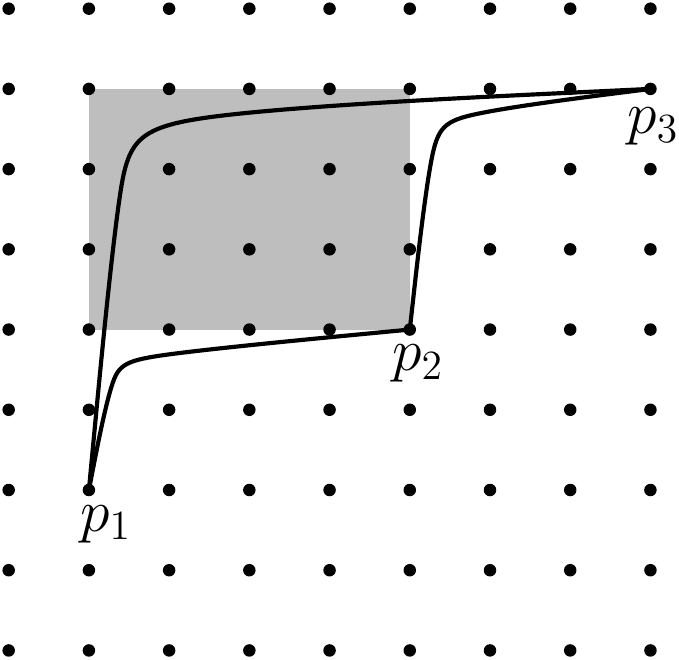}
\caption{A grid and three points $p_1=(a_1,b_1)$, $p_2=(a_2,b_2)$ and $p_3=(a_3,b_3)$ with $a_1<a_2<a_3$ and $b_1<b_2<b_3$. The Jordan arcs between the points represent straight-line segments between the stretched points $\tau(a_1,b_1)$, $\tau(a_2,b_2)$ and $\tau(a_3,b_3)$.
The rectangle $\Box(\Delta)$ defined for the triangle for $\Delta=\Delta(p_1,p_2,p_3)$ is shaded.}
\label{fig:stretching}
\end{figure}

For each triangle $\Delta$ determined by $(a_1,b_1)$, $(a_2,b_2)$, and $(a_3,b_3)$, we define an open axis-aligned rectangle $\Box(\Delta)$. Assume, by permuting the indices if necessary, that $\max(b_1,b_2)< b_3$, and let
$$\Box(\Delta)=(\min(a_1,a_2),\max(a_1,a_2))\times (\max(b_1,b_2), b_3).$$

\begin{lemma}\label{lem:rectangle}
Let $(a_1,b_1), (a_2,b_2), (a_3,b_3)\in A_n$ such that $\max(b_1,b_2)< b_3$.
Then for all $p\in A_n\cap \Box(\Delta)$, the point $\tau(p)$ lies in the
interior of the triangle determined by $\tau(p_1)$, $\tau(p_2)$, and $\tau(p_3)$.
\end{lemma}
\begin{proof}
By Lemma~\ref{lem:strechedsegment}, $\tau(a_3,b_3)$ all points $\tau(p)$, $p\in A_n\cap \Box(\Delta)$, lie above the line spanned by $\tau(a_1,b_1)$ and $\tau(a_2,b_2)$. Consider now the position of $a_3$ relative to $a_1$ and $a_2$.
Assume, without loss of generality, that $a_1<a_2$ (if $a_1=a_2$, then $\Box(\Delta)=\emptyset$, and the conclusion trivially holds.)

If $a_3=a_1$, then $\tau(a_1,b_1)$ and $\tau(a_3,b_3)$ span a vertical line, and the point $\tau(a_2,b_2)$ and all points $\tau(p)$, $p\in A_n\cap \Box(\Delta)$, lie to the right of this line. If $a_3<a_1$ (resp., $a_3>a_1$), then $\tau(a_2,b_2)$ and all points $\tau(p)$, $p\in A_n\cap \Box(\Delta)$, lie above (resp., below) the line spanned by $\tau(a_1,b_1)$ and $\tau(a_3,b_3)$ by Lemma~\ref{lem:strechedsegment}.

Similarly, If $a_3=a_2$, then $\tau(a_2,b_2)$ and $\tau(a_3,b_3)$ span a vertical line, and the point $\tau(a_1,b_1)$ and all points $\tau(p)$, $p\in A_n\cap \Box(\Delta)$, lie on the left of this line. If $a_3<a_2$ (resp., $a_3>a_2$), then $\tau(a_1,b_1)$ and all points $\tau(p)$, $p\in A_n\cap \Box(\Delta)$, lie below (resp., above) the line spanned by $\tau(a_2,b_2)$ and $\tau(a_3,b_3)$ by Lemma~\ref{lem:strechedsegment}. In all cases, all points $\tau(p)$, $p\in A_n\cap \Box(\Delta)$, lie in the interior of $\Delta$.
\end{proof}

\paragraph{Grid-embedding in a stretched grid.}
The grid-embedding algorithm by de~Fraysseix~et al.~\cite{FPP90} embeds every $n$-vertex planar graph on an $(2n-4)\times (n-2)$ section of the integer lattice. Their algorithms also works on the stretched grid in place of the integer grid.

Specifically, we use their result in the following form. Suppose that $G_m$ is a planar graph with $m\in \mathbb{N}$ vertices and endowed with a given combinatorial embedding in which $u$, $v$ and $z$ are the vertices of the outer face. Let $X,Y\subset \mathbb{N}$ be two sets of cardinality $|X|\geq 2m$ and $|Y|\geq m$. Then $G$ has a straight-line embedding such that the vertices are mapped to the stretched cross product $\tau(X\times Y)$ of size at least $2m^2$; the two endpoints of edge $uv$ are mapped to $\tau(\min X,\min Y)$ and
$\tau(\max X,\min Y)$, respectively; and $z$ is mapped to an arbitrary point
in the top row $\tau(X\times \max Y)$.
Furthermore, By Lemma~\ref{lem:strechedsegment}, we can shift $u$ or $v$ vertically down to another point of the stretched grid (while keeping all other vertices fixed) without introducing any edge crossings. Similarly, we can shift $z$ horizontally to any other point
of the stretched grid without introducing any edge crossings.

\begin{lemma}\label{lem:gridembed}
Let $G_m$ be a planar 3-tree with $m\in\mathbb{N}$ vertices and a combinatorial embedding  in which $u$, $v$ and $z$ are the vertices of the outer face.
Let $p_1, p_2, p_3\in B_n$, and let $R\subseteq \Box(\Delta(p_1,p_2,p_3))$ be a
rectangle such that $\area(R)> 8 m^2 n^\alpha$, $w(R)> 2m$, and $h(R)> m$.
Then $G_m$ admits a planar straight-line embedding such that $u$, $v$, and $z$
are mapped to $p_1$, $p_2$, and $p_3$, respectively, and the interior vertices of
$G_m$ are mapped to points $\tau(p)$, $p\in B_n\cap R$.
\end{lemma}
\begin{proof}
It is enough to show that $B_n\cap R$ contains a cross
product $X\times Y$ such that $|X|\geq 2m$ and $|Y|\geq m$.
Recall that $n^\alpha=2^q$ for some $q\in \mathbb{N}$.
By decreasing the width and height of $R$, we obtain
a rectangle $R'\subseteq R$ with
width $w(R')> 2^k(2m)$, height $h(R')> 2^{q-k}m$, and
$\area(R')> 2^q(2m^2)$, for some integer $k$, \
$0\leq k\leq q$. The points $(i,j)\in R'$, with $i\equiv 0\mod 2^k$
and $j\equiv 0\mod 2^{q-k}$ are in $B_n\cap {\rm int}(R)$ and form a
required cross product $X\times Y$ such that $|X|\geq 2m$ and $|Y|\geq m$.

We can now embed $u$, $v$, and $z$ at $p_1$, $p_2$, and $p_3$, respectively.
The algorithm by de~Fraysseix~et al.~\cite{FPP90} embeds the interior vertices
of $G_m$ to points $\tau(p)$, $p\in X\times Y\subset B_n\cap R$,
as required.
\end{proof}

\section{Embedding Algorithm\label{sec:alg3}}

In this section, we show that for $\alpha=1/2$, every $n$-vertex planar 3-tree admits a straight-line embedding such that the vertices are mapped to the set $S_n$ of size $\Theta(n^{3/2}\log n)$.

\paragraph{Overview.}
Let $G$ be a planar 3-tree with $n$ vertices. We describe our embedding algorithm in term of the ``unstretched'' grid $B_n$. The function $\tau$ maps this embedding into a straight-line embedding into $S_n$. Our embedding algorithm is guided by the tree $T=T(G)$, which represents an incremental process that constructs $G$ from a single triangle. Recall that $T_\Delta$ denotes the subtree of $T$ rooted at a node $\Delta$; and $V_\Delta$ denotes the set of vertices of $G$ that correspond to nodes in $T_\Delta$.

Our algorithm \emph{processes} the nodes $\Delta\in V(T)$ in a breath-first traversal of $T$. When a triangle $\Delta$ is already embedded in the point set $S_n$, then the rectangle $\Box(\Delta)$ is well defined, the vertices $V_\Delta$ are mapped to points in $B_n\cap \Box(\Delta)$. In order to maintain additional properties (invariant~$I_2$ below), we also maintain an open rectangle $R(\Delta)\subseteq \Box(\Delta)$, and require that the vertices $V_\Delta$ be mapped into $B_n\cap R(\Delta)$. Intuitively, $R(\Delta)$ is the region ``allocated'' for the vertices in $V_\Delta$.

\begin{figure}[htb]
\centering
\includegraphics[width=0.75\textwidth]{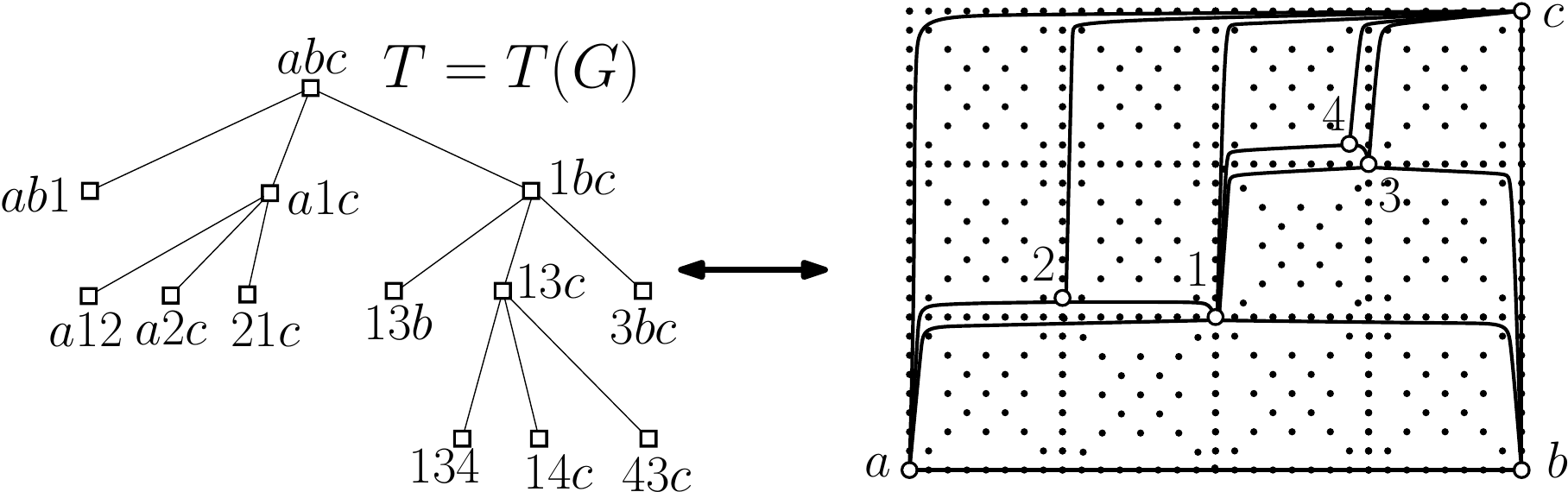}
\caption{The embedding of a 3-tree $G$ from Fig.~\ref{fig:tree3} on a sparse grid.}
\label{fig:vertexembedding3}
\end{figure}

When the breath-first traversal of $T$ reaches a node $\Delta\in V(T)$ such that $\area(R(\Delta))> 8 n^\alpha \weight^2(\Delta)$, $w(R(\Delta))> 2\weight(\Delta)$, and $h(R(\Delta))> \weight(\Delta)$, then we complete the embedding of the vertices $V_\Delta$ by Lemma~\ref{lem:gridembed}. We call the set of nodes of $T$ where these conditions are first satisfied the \emph{fringe} of $T$. We show below that Lemma~\ref{lem:gridembed} becomes applicable by the time $\weight(\Delta)$ drops below $n^{1-\alpha}$. For nodes $\Delta\in V(T)$ below the fringe, there is no need to assign rectangles $R(\Delta)$.

When we process a node $\Delta\in V(T)$ that is a parent of a hub, we shift some of the previously embedded vertices in horizontal or vertical direction (as described below), and shift the corers of rectangles $R(\Delta')$ corresponding to previously processed nodes $\Delta'\in V(T)$, as well. Each shift operation changes the $x$- or $y$-coordinate of a point (with respect to the ``unstretched'' point set $B_n$) by $0$, $n^\alpha$, or $2n^\alpha$. The number of hubs is at most $2n^{1-\alpha}$ by Lemma~\ref{lem:hubs}, so the $x$- and $y$-coordinate of each vertex may be shifted by at most $2n^{1-\alpha}\cdot 2n^\alpha=4n$. To allow sufficient space for these operations, we initially start with a $10n\times 10n$ section of the sparse grid $B_n$, and the shift operations may expand the bounding box to up to $14n\times 14n$.

\paragraph{Invariants.}
For all nodes $\Delta\in V(T)$ on or above the fringe of $T$, we maintain the following invariants.
\begin{enumerate}\itemsep -2pt
\item[$I_1$]
$R(\Delta)\subseteq \Box(\Delta)$.
\item[$I_2$]
If $\weight(\Delta)\geq n^{\alpha}$, then the lower-left corner of $R(\Delta)$ is in a forward diagonal,
and the lower-right corner of $R(\Delta)$ is in a backward diagonal of $B_n$.
\item[$I_3$]
At least one of the following two conditions is satisfied:
\begin{eqnarray}
\area(R(\Delta)) &\geq &100 n\weight(\Delta);\label{eq:I3}\\
\area(R(\Delta)) &>& 8 n^\alpha\weight^2(\Delta), w(R(\Delta))> 2\weight(\Delta),
\mbox{ \rm and }h(R(\Delta))> \weight(\Delta).\label{eq:I4}
\end{eqnarray}
\end{enumerate}
Condition \eqref{eq:I4} implies that $\Delta$ is on the fringe, and the vertices $V_\Delta$ can be embedded by Lemma~\ref{lem:gridembed}.

\paragraph{Initialization.}
Denote by $abc$ the initial triangle of $G$, with $a$ labeled left, $b$ labeled right and $c$ labeled top. Then we have $T=T_{abc}$. Let $R(abc)$ be the interior of bounding box of a $10n\times 10n$ section of $B_n$. Embed $a$ and $b$ to the lower-left and lower-right corners of $R(abc)$, respectively. Embed $c$ in the upper-right corner of $R(abc)$ (see Fig.~\ref{fig:vertexembedding3}). It is clear that invariants $I_1$--$I_3$ are satisfied for $abc$.

\paragraph{One recursive step.}
Assume that the vertices of triangle $\Delta\in V(T)$ have already been embedded and we are given a rectangle $R(\Delta)$ satisfying invariants $I_1$--$I_3$. If $\Delta$ is on the fringe of $T$, then the embedding of the vertices $V_\Delta$ is completed by Lemma~\ref{lem:gridembed}, and the subtree $T_\Delta$ is removed from further consideration.

In the remainder of this section, we assume that node $\Delta$ is strictly above the fringe, where invariant~$I_3$ is satisfied with \eqref{eq:I3} rather than \eqref{eq:I4}. Note that $w(R(\Delta))\leq 14n$ and $h(R(\Delta))\leq 14n$ since the bounding box of $B_n$ is a $14n\times 14n$ square. These upper bounds, combined with \eqref{eq:I3}, yield the following lower bounds for heavy nodes:
\begin{equation}\label{eq:I33}
w(R(\Delta))> 7\weight(\Delta)
\hspace{1cm} \mbox{ \rm and } \hspace{1cm}
h(R(\Delta))> 7\weight(\Delta).
\end{equation}
Since $\Delta$ does not satisfy \eqref{eq:I4} despite \eqref{eq:I33}, we have $\area(R(\Delta))\leq 8n^\alpha\weight^2(\Delta)$. This, combined with \eqref{eq:I3}, yields
\begin{equation}\label{eq:heavy}
\weight(\Delta)\geq 12.5 n^{1-\alpha}.
\end{equation}

\paragraph{Ideal location for a vertex.}
Denote the bottom, left and right child of $\Delta$, respectively, by $\Delta_1$, $\Delta_2$ and $\Delta_3$. Suppose that $R(\Delta)= (a,b)\times (c,d)$.
We wish to place the vertex $v$ corresponding to $\Delta$ at some point $p\in B_n\cap R(\Delta)$. The point $p\in B_n\cap R(\Delta)$ subdivides $R(\Delta)$ into a bottom, left, and right rectangle, $\Box(\Delta_1)\cap R(\Delta)$, $\Box(\Delta_2)\cap R(\Delta)$, and $\Box(\Delta_3)\cap R(\Delta)$ corresponding to $\Delta_1$, $\Delta_2$, and $\Delta_3$. We choose an ``ideal'' location for $v$ (which may not a point in $B_n$) that would ensure that the area of $R(\Delta)$ is distributed among the rectangles $R(\Delta_i)$, $i=1,2,3$, proportionally to $\weight(\Delta_i)$, maintaining \eqref{eq:I3}. Refer to Fig.~\ref{fig:snapping-0}. Recall that
$$\weight(\Delta)=\weight(\Delta_1)+\weight(\Delta_2)+\weight(\Delta_3)+1.$$
Partition the area of $R(\Delta)$ into a top and a bottom part by a horizontal line $\ell_y$ in ratio
$$\left(\weight(\Delta_2)+\weight(\Delta_3)+\frac{2}{3}\right)
: \left(\weight(\Delta_1)+\frac{1}{3}\right).$$
Partition the top area of $R(\Delta)$ into a left and a right part by a vertical line $\ell_x$ in ratio
$$\left(\weight(\Delta_2)+\frac{1}{3}\right) : \left(\weight(\Delta_3)+\frac{1}{3}\right).$$
The ideal location for vertex $v$ is the intersection point $\ell_x\cap \ell_y$. Note that the ideal location is in the interior of $R(\Delta)$ even if $\weight(\Delta_1)$, $\weight(\Delta_3)$, or $\weight(\Delta_3)$ is 0. If we place vertex $v$ at $\ell_x\cap \ell_y$, then $R(\Delta)$ is partitioned among the rectangles $R(\Delta_i)$, $i=1,2,3$, such that
\begin{equation}\label{eq:area}
\area(R(\Delta_i))\geq \area(R(\Delta))\frac{\weight(\Delta_i)}{\weight(\Delta)}.
\end{equation}
It is clear that if $\Delta$ satisfies~\eqref{eq:I4} and its children satisfy \eqref{eq:area} for $i=1,2,3$,
then the children also satisfy \eqref{eq:I4}.

Unfortunately, the ideal location is not necessarily in $B_n$ (and often not in $A_n$). Inevitably, some of the constraints have to be relaxed. We distinguish several cases: When one of the children of $\Delta$ is a hub, we snap vertex $v$ to a point in $B_n$ near the ideal location, and shift some of the previously embedded vertices (as explained below) to restore the lower bounds in \eqref{eq:area}. When none of the children of $\Delta$ is a hub, then two or three of its children are light by Lemma~\ref{lem:sibling}: we maintain invariant $I_3$ by establishing either \eqref{eq:I3} or \eqref{eq:I4} for every child of $\Delta$.

\paragraph{Case~1: one of the children of $\Delta$ is a hub.}
We wish to embed the vertex $v$ corresponding to $\Delta$ so that invariant~$I_2$ is maintained for all children of $\Delta$ (at least one of them is a hub). We will shift the previously embedded vertices and rectangles by $0$, $n^\alpha$ or $2n^\alpha$, and embed the vertex $v$ corresponding to $\Delta$ into a point ``near'' its ideal location.

\begin{figure}[htb]
\centering
\subfloat[]{
\label{fig:snapping-0}
\includegraphics[height=1.7in]{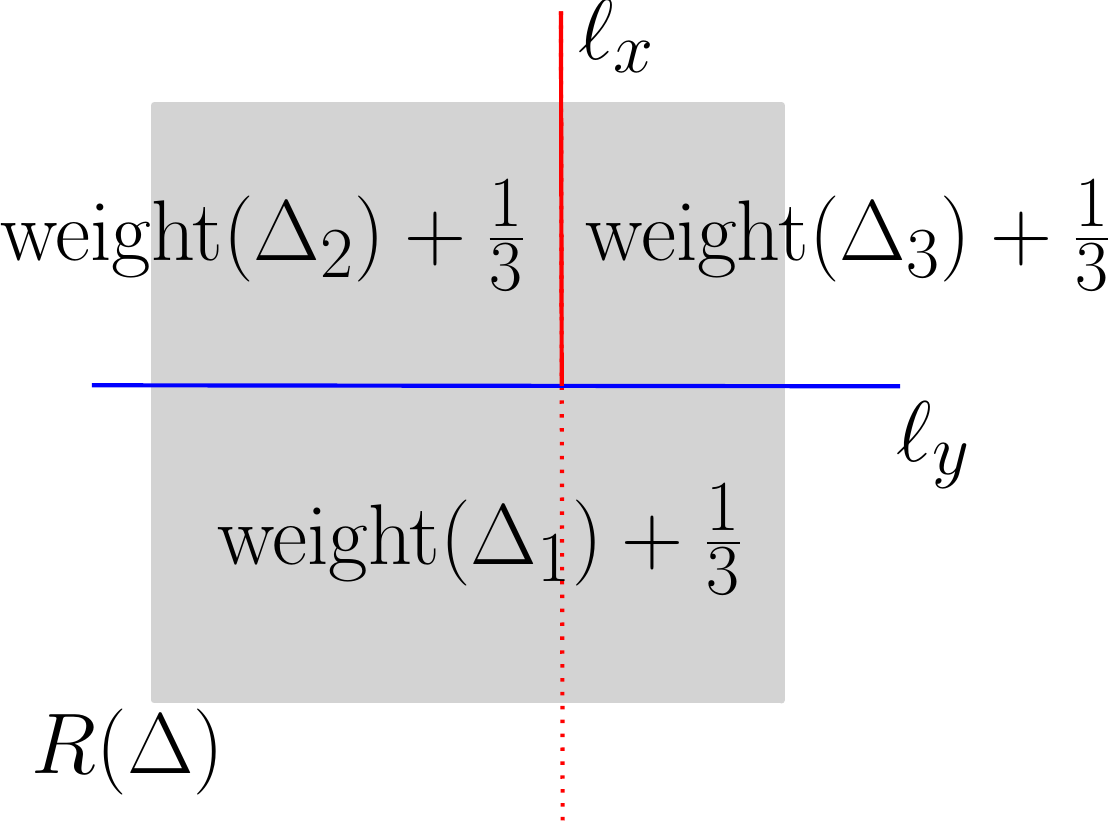}
}
\hspace{2cm}
\subfloat[]{
\label{fig:snapping-a}
\includegraphics[height=1.8in]{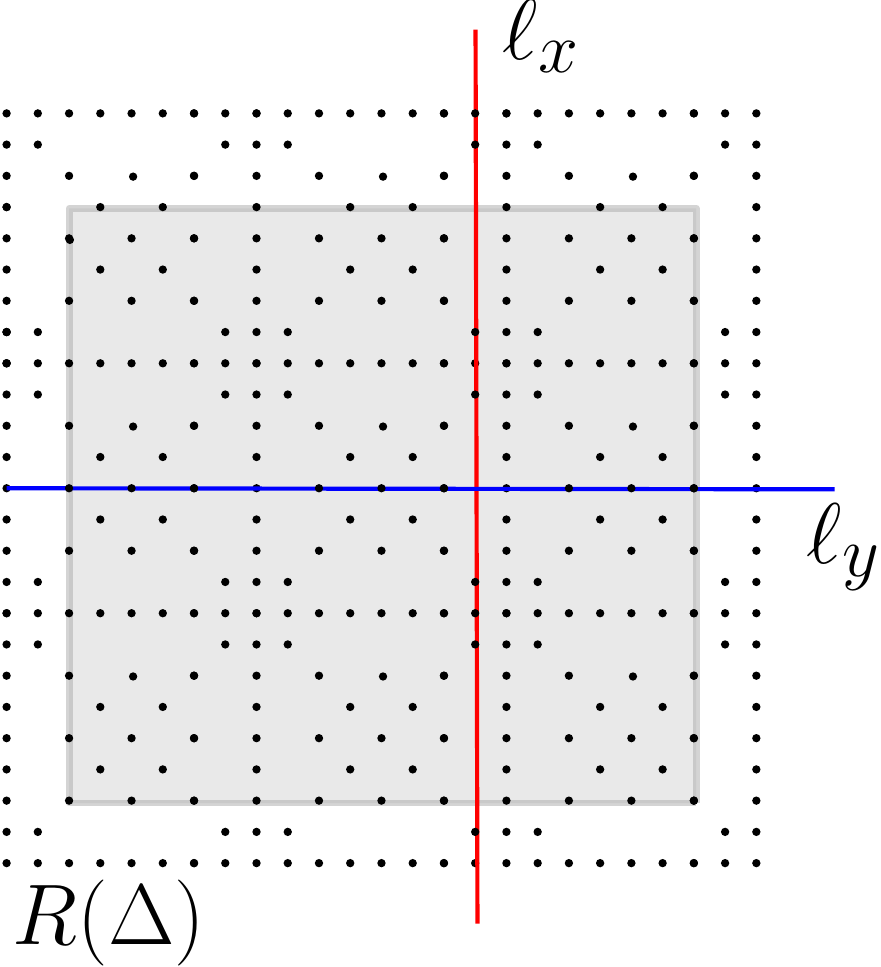}
}
\vspace{\baselineskip}
\subfloat[]{
\label{fig:snapping-b}
\includegraphics[height=2in]{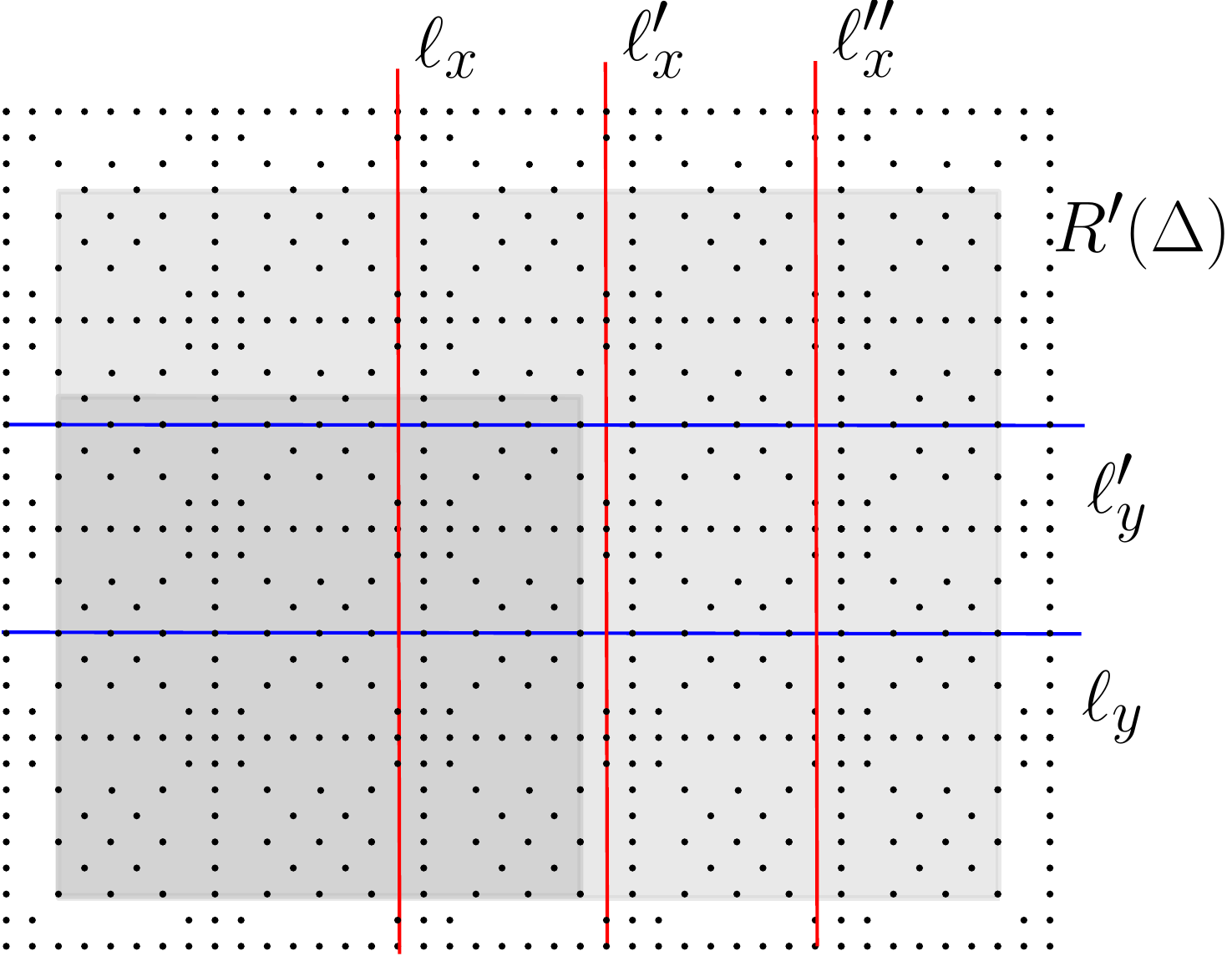}
}
\hspace{1cm}
\subfloat[]{
\label{fig:snapping-c}
\includegraphics[height=2in]{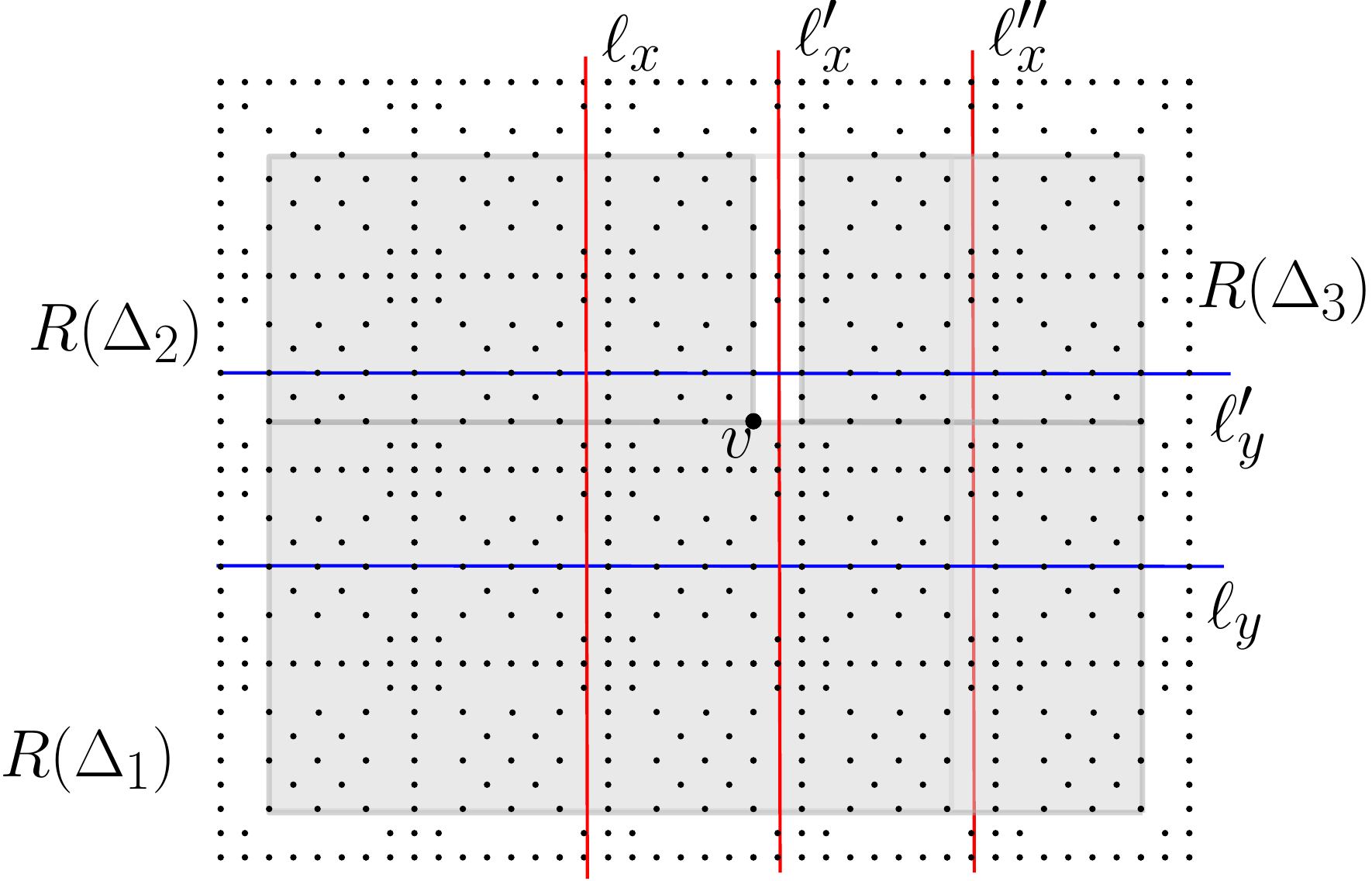}
}
\caption{
(a) The lines $\ell_x$ and $\ell_y$ partition $\area(R(\Delta))$ proportionally to the weights $\weight(\Delta_i)+\frac{1}{3}$, $i=1,2,3$.
(b) The lines $\ell_x$ and $\ell_y$ define the ideal location of a vertex $v$.
(c) Shift every corner on the right of $\ell_x$ by $2n^{\alpha}$ to the right, and every corner above $\ell_y$ by $n^{\alpha}$ up.
(d) The location of  vertex $v$, and the rectangles $R(\Delta_1)$, $R(\Delta_2)$, and $R(\Delta_3)$.\label{fig:snapping}}
\end{figure}

Let $\ell_x$ and $\ell_y$ be the vertical and horizontal lines that define the ideal location for $v$ (Fig.~\ref{fig:snapping-a}). By definition, both intersect the interior of rectangle $R(\Delta)$. For all previously embedded vertices and all corners of previously defined rectangles $R(.)$, if they lie on or to the right of  of line $\ell_x$, then increase their $x$-coordinates by $2n^{\alpha}$; if they lie on or above the line $\ell_y$, then increase their $y$-coordinates  by $2n^{\alpha}$. Note that the width (resp., height) of all previously defined rectangles increases by 0 or $2n^\alpha$ (resp., 0 or $n^\alpha$). In particular, $R(\Delta)$ is expanded to a rectangle $R'(\Delta)$ whose width and height are $w(R(\Delta))+2n^\alpha$ and $h(R(\Delta))+n^\alpha$, respectively. Let $\ell_y'$ be a horizontal line at distance $n^\alpha$ above $\ell_y$; and let $\ell_x'$ and $\ell_x''$ be two vertical lines at distance $n^\alpha$ and $2n^\alpha$ to the right of $\ell_x$ (Fig.~\ref{fig:snapping-b}).

We now embed vertex $v$ at a point $B_n\cap R'(\Delta)$, and define rectangles $R(\Delta_i)$ for $i=1,2,3$. We shall choose the rectangles $R(\Delta_i)$, $i=1,2,3$, such that their widths and heights are at least as as large as if $v$ were placed at the ideal location in $R(\Delta)$. In addition, we also ensure that the lower left (resp., lower right) corner of
$R(\Delta_i)$, $i=1,2,3$, lies on a forward diagonal (resp., backward diagonal), thereby establishing invariant $I_2$ for all children of $\Delta$.

We define $R(\Delta_i)$, $i=1,2,3$, as follows (refer to Fig.~\ref{fig:snapping-c}). The bottom child of $\Delta$ is $\Delta_1$. Let the bottom side of $R(\Delta_i)$ be the bottom side of $R'(\Delta)$, and let its top side be the unique segment between $\ell_y$ and $\ell_y$ such that the height of $R(\Delta_i)$ is a multiple of $n^\alpha$. Since the lower left (right) corner of both $R(\Delta)$ and $R'(\Delta)$ are on a forward (backward) diagonal by invariant $I_2$, this is also true for $R(\Delta_1)$. Since the height of $R(\Delta_1)$ is a multiple of $n^\alpha$, the upper left (right) corner of $R(\Delta_1)$ is also on a forward (backward) diagonal.

Let the lower left corner of $R(\Delta_2)$ be the upper left corner of $R(\Delta_1)$, and choose $x$-coordinate of the lower right corner of $R(\Delta_2)$ between lines $\ell_x$ and $\ell_x'$ such that it is on a backward diagonal. Similarly, let the lower right corner of $R(\Delta_3)$ be the upper right corner of $R(\Delta_1)$, and choose the $x$-coordinate of its lower left corner between lines $\ell_x'$ and $\ell_x''$ such that it is on a forward diagonal. Let top side of both $R(\Delta_2)$ and $R(\Delta_3)$ be part of the top side of $R'(\Delta)$. Finally, embed vertex $v$ at the lower right corner of $R(\Delta_2)$.

This ensures $R(\Delta)\subseteq \Box(\Delta_i)\cap R'(\Delta)$ for $i=1,2,3$, maintaining invariant $I_1$. Note that the width and height of the rectangle $R(\Delta_i)$ are at least as large as if $v$ were placed at the ideal location within in $R(\Delta)$, establishing \eqref{eq:I33} hence \eqref{eq:I3} for $i=1,2,3$.
Invariants $I_1$--$I_3$ are maintained for $\Delta_1$, $\Delta_2$, and $\Delta_3$.

\paragraph{Preliminaries for Cases~2 and~3.}
In the remaining cases, the children of $\Delta$ are not hubs. By Lemma~\ref{lem:sibling},  $\Delta$ has at most one heavy child. For the one possible heavy child $\Delta_i$, $i\in \{1,2,3\}$, we shall choose a rectangle $R(\Delta_i)$ satisfying \eqref{eq:I3}; and we establish \eqref{eq:I4} for two light children of $\Delta$.

\paragraph{Case 2: the children of $\Delta$ are not hubs, and the left and right children are light.}
In this case, we use the following strategy (refer to Fig.~\ref{fig:case2+}). We choose pairwise disjoint ``preliminary'' rectangles $R_0(\Delta_1)$, $R_0(\Delta_2)$, and $R_0(\Delta_3)$ in $R(\Delta)$ such that $R_0(\Delta_1)$ satisfies \eqref{eq:I3}; and $R_0(\Delta_2)$ and $R_0(\Delta_3)$ satisfy \eqref{eq:I4}.
We also choose a rectangular region $Q\subset R(\Delta)$ such that placing $v$ at any point in $Q$ yields $R_0(\Delta_i)\subset \Box(\Delta_i)$, for $i=1,2,3$. Finally, we show that $Q$ contains a point from a full column. We place $v$ at an arbitrary point in $Q\cap B_n$, and put $R(\Delta_i)=\Box(\Delta_i)\cap R(\Delta)$.

Let $\delta_1=\weight(\Delta_2)+\weight(\Delta_3)+2/3$. Recall that the horizontal line $\ell_y$ partitions the area of $R(\Delta)$ into a top and a bottom part in ratio $\delta_1: (\weight(\Delta_1)+1/3)$. Decompose the top part of $R(\Delta)$ into 6 congruent rectangles by a horizontal line and two vertical lines. Now, let $R_0(\Delta_2)$ and $R_0(\Delta_3)$ be the upper left and upper right congruent rectangles, respectively, and let $Q$ be the lower middle rectangle (see Fig.~\ref{fig:case2+}).

\begin{figure}[htb]
\centering
\subfloat[]{
\label{fig:case2+}
\includegraphics[height=1.8in]{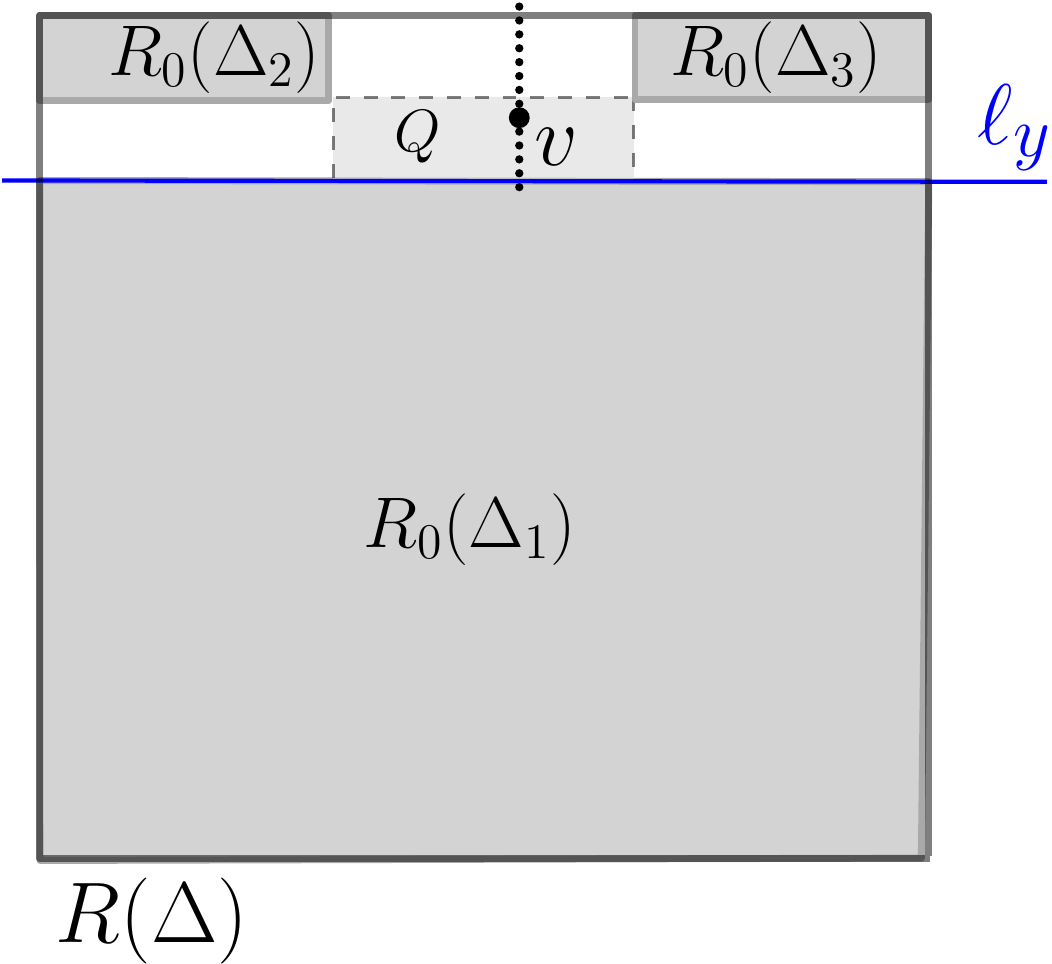}
}
\hspace{2cm}
\subfloat[]{
\label{fig:case3a}
\includegraphics[height=1.8in]{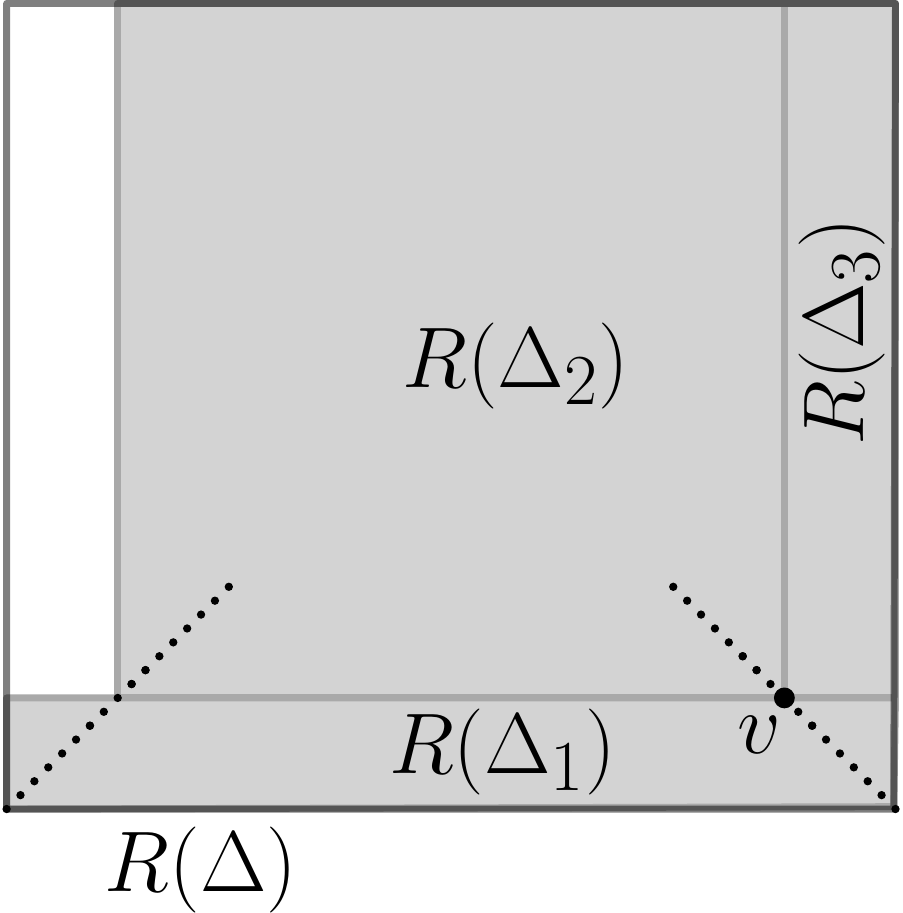}
}
\caption{
(a) When the left and right children are light (Case~2), vertex $v$ is placed on a full column.
(b) When the left or right child is heavy (Case~3), vertex $v$ is placed on a forward or backward diagonal.}
\label{fig:cases}
\end{figure}

The height $h(Q)=h(R_0(\Delta_2))=h(R_0(\Delta_3))$ of the congruent rectangles is bounded by
$$h(Q)=\frac{h(R(\Delta))}{2}\cdot \frac{\delta_1}{\weight(\Delta)}
     \geq \frac{7\weight(\Delta)}{2} \cdot \frac{\delta_1}{\weight(\Delta)}
     = \frac{7\delta_1}{2}\geq \frac{7(2/3)}{2}=\frac{7}{3}.$$
Their width is $w(Q)=w(R_0(\Delta_2))=w(R_0(\Delta_3))=\frac{1}{3}w(R(\Delta))>29n^{1-\alpha}$,
using \eqref{eq:I33} and \eqref{eq:heavy}. Hence $w(Q)\geq n^\alpha$ when $\alpha\leq 1/2$, and so $Q$
contains a point on a full column of $B_n$.

We place $v$ at an arbitrary point in $Q\cap B_n$, and put $R(\Delta_i)=\Box(\Delta_i)\cap R(\Delta)$ for $i=1,2,3$. The lower left (resp., lower right) corner of $R(\Delta_1)$ is the same as that of $R(\Delta)$, which establishes invariant $I_2$ for $\Delta_1$. Note that $\Delta_1$ satisfies \eqref{eq:I3}, since $\area(R(\Delta_1))\geq \area(R_0(\Delta_1))$. We show that $\Delta_2$ and $\Delta_3$ satisfy \eqref{eq:I4}. For $i=2,3$, we have $w(R(\Delta_i))\geq w(R_0(\Delta_i))\geq 29n^{1-\alpha} \geq n^\alpha\geq \weight(\Delta_i)$;  $h(R(\Delta_i))\geq h(R_0(\Delta_i))\geq \frac{7}{3}\delta_1\geq \weight(\Delta_i)$; and $\area(R(\Delta_i))$ is bounded by
\begin{eqnarray}
\area(R(\Delta_i))
&=& w(R(\Delta_i))h(R(\Delta_i))
\geq \frac{w(R(\Delta))}{3}\cdot \frac{h(R(\Delta))\delta_1}{\weight(\Delta_1)+1/3}
\geq \frac{\area(R(\Delta))\delta_1}{3\weight(\Delta)}\nonumber\\
&\geq& \frac{100n\weight(\Delta)\delta_1}{3\weight(\Delta)}
> 33n\delta_1
> 33n\weight(\Delta_i)
\geq 33n^{1-\alpha}\weight^2(\Delta_i),\nonumber
\end{eqnarray}
which is more than $n^\alpha\weight^2(\Delta_i)$ when $\alpha\leq 1/2$.

\paragraph{Case 3: the children of $\Delta$ are not hubs, and the left or right child of $\Delta$ is heavy} Assume that $\Delta_2$ is heavy (the case that $\Delta_3$ is heavy is treated analogously). Let $\delta_2=\weight(\Delta_1)+\weight(\Delta_3)+1$. We distinguish between two possibilities.

\paragraph{Case~3A: $4\delta_2 n^\alpha\leq w(R(\Delta))$ and $4\delta_2 n^\alpha \leq h(R(\Delta))$} Refer to Fig.~\ref{fig:case3a}. Place vertex $v$ corresponding to $\Delta$ at $(c+2\delta_2,b-2\delta_2)$ on a forward diagonal, and assign the rectangles
$R(\Delta_1)=(a,b)\times (c,c+2\delta_2)$, $R(\Delta_2):=(a+2\delta_2,b-2\delta_2)\times (c+2\delta_2,d)$; and $R(\Delta_3)=(d-2\delta_2,d)\times (c+2\delta_2,d)$.
Note that the lower-left (resp., lower-right) corner of $R(\Delta_2)$ is on a forward (resp., backward) diagonal of $B_n$, establishing invariant $I_3$.

By construction, rectangles $R(\Delta_i)$, $i=1,2,3$, satisfy invariants $I_1$ and $I_2$.
We establish \eqref{eq:I3} for $\Delta_2$:
\begin{eqnarray}
\area(R(\Delta_2))
&=&   [w(R(\Delta))-4\delta_2]\cdot [h(R(\Delta)) -2\delta_2]
\geq  \area(R(\Delta)) - [2w(\Delta)+4h(\Delta)]\delta_2\nonumber\\
&\geq&100n\weight(\Delta) - 6\cdot 14n\delta_2
\geq 100n(\weight(\Delta)-\delta_2)
= 100n\weight(\Delta_2).\nonumber
\end{eqnarray}
For $\Delta_1$ and $\Delta_3$, we establish \eqref{eq:I4}.
We have $w(R(\Delta_1))=w(R(\Delta))\geq 4\delta_2 n^\alpha \geq 4\weight(\Delta_1)n^\alpha$ and $h(R(\Delta_3)) = h(R(\Delta)) - 2\delta_2  \geq (4\delta_2 -2\delta_2/n^\alpha)n^\alpha\geq (4\weight(\Delta_3)-2)n^\alpha$.
On the other hand, $h(R(\Delta_1))=w(R(\Delta_3))= 2\delta_2\geq 2\weight(\Delta_1)+2,2\weight(\Delta_3)+2$ by construction.
%
%This establishes \eqref{eq:I4} when $\weight(\Delta_1)>1$ and $\weight(\Delta_3)>1$.
%Otherwise, it is easy to see that we have $\area(R(\Delta_1))>n^\alpha$ if  $\weight(\Delta_1)=1$ and $\area(R(\Delta_3))>n^\alpha$ if
%$\weight(\Delta_3)=1$.

%Finally, $h(\Delta_2)\geq h(\Delta_3) >7\weight(\Delta) >28n^{1-\alpha}>14\delta_2$, whenever $\alpha\leq 1/2$.

\begin{figure}[htb]
\centering
\subfloat[]{
\label{fig:case3b-}
\includegraphics[height=1.8in]{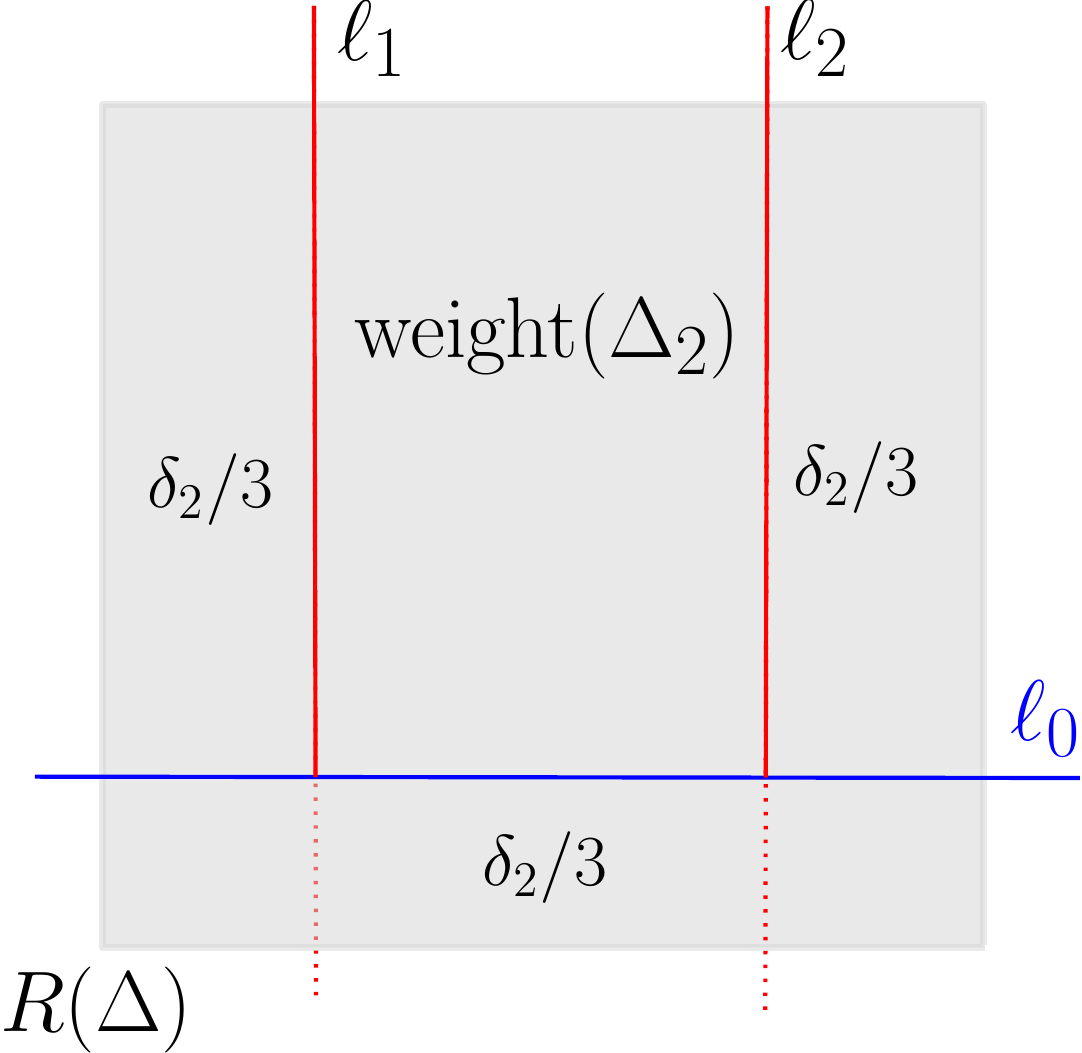}
}
\hspace{4mm}
\subfloat[]{
\label{fig:case3b}
\includegraphics[height=1.8in]{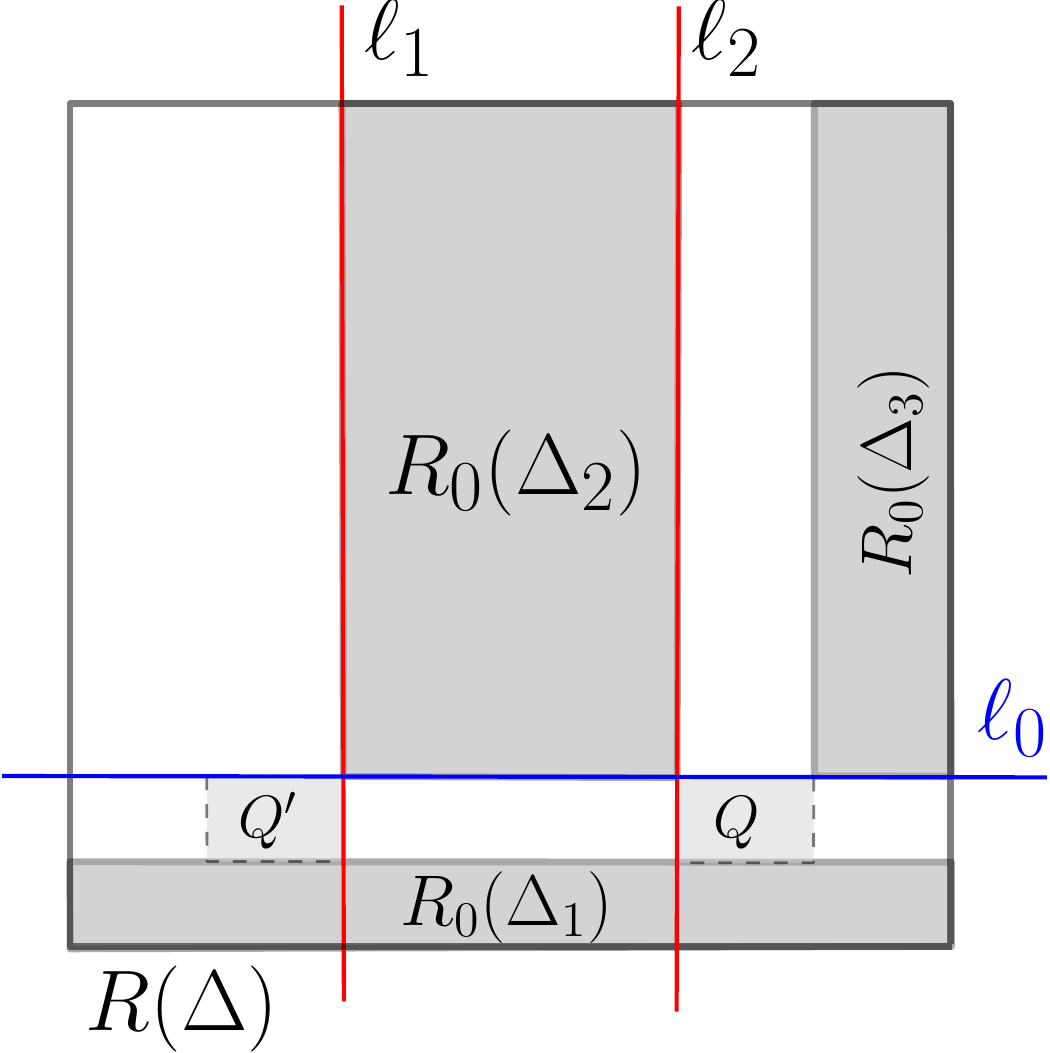}
}
\hspace{4mm}
\subfloat[]{
\label{fig:case3b+}
\includegraphics[height=1.8in]{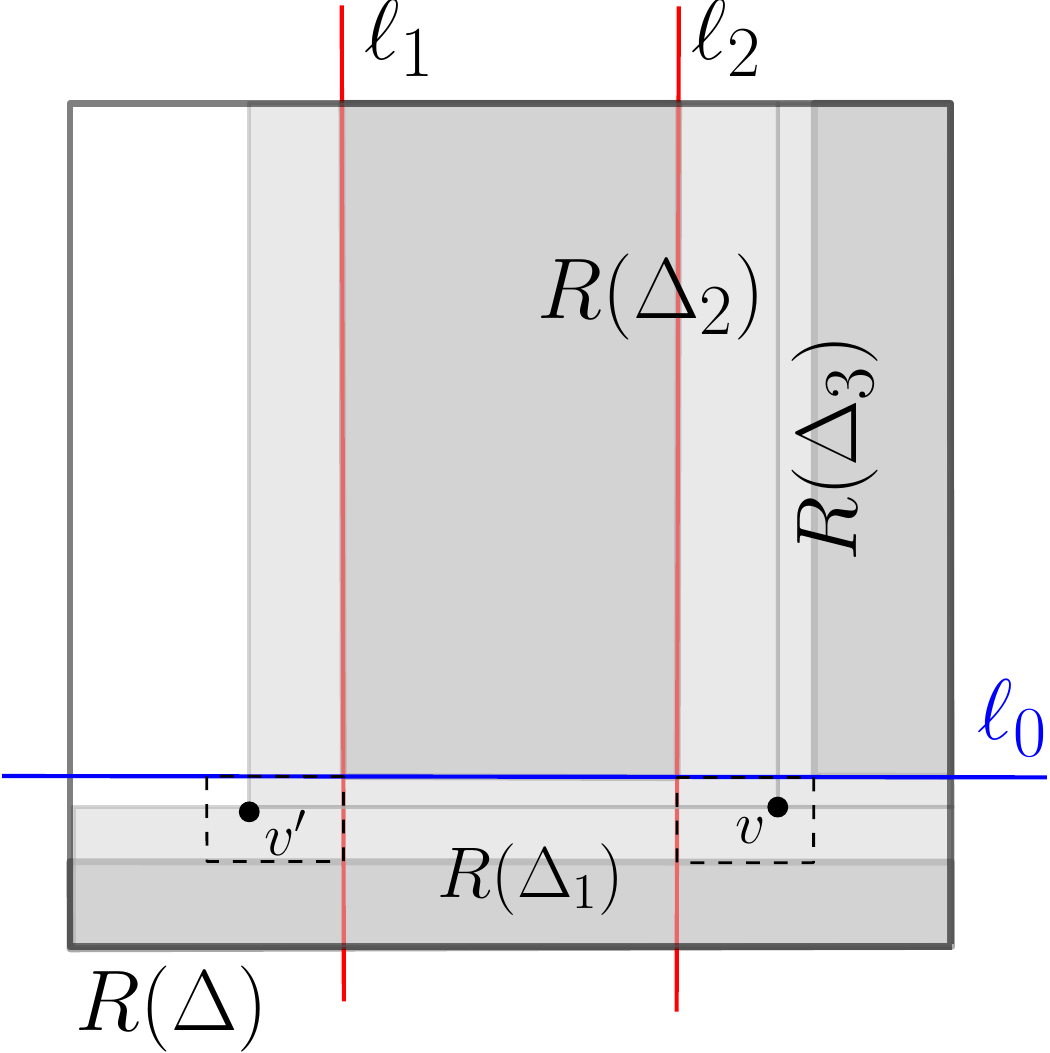}
}
\caption{
(a) The lines $\ell_0$, $\ell_1$, and $\ell_2$ partition the area of $R(\Delta)$ proportionally to $\delta_2/3$, $\weight(\Delta_2)$, $\delta_2/3$, and $\delta_2/3$.
(b) We define ``preliminary'' rectangles $R_0(\Delta_i)$, $i=1,2,3$, and a rectangle $Q$.
(c) Vertex $v$ is placed at a point in $B_n\cap Q$ lying on a backward diagonal.
}
\label{fig:Q}
\end{figure}

\paragraph{Case~3B: $4\delta_2 n^\alpha> w(R(\Delta))$ or $4\delta_2 n^\alpha >h(R(\Delta))$}
In this case, we follow a strategy similar to Case~2 (refer to Fig.~\ref{fig:case3b}). We choose pairwise disjoint ``preliminary'' rectangles $R_0(\Delta_1)$, $R_0(\Delta_2)$, and $R_0(\Delta_3)$ in $R(\Delta)$ such that $R_0(\Delta_2)$ satisfies \eqref{eq:I3}; and $R_0(\Delta_1)$ and $R_0(\Delta_3)$ satisfy \eqref{eq:I4}. We also choose a rectangular region $Q\subset R(\Delta)$ such that placing vertex $v$ at any point in $Q$ yields $R_0(\Delta_i)\subset \Box(\Delta_i)$, for $i=1,2,3$. Finally, we show that $Q$ contains a point on backward diagonal by Lemma~\ref{lem:dense2}. We place $v$ at such a point $p\in Q$. Choose $R(\Delta_3)$ such that $R_0(\Delta_3)\subseteq R(\Delta_3)$, its lower right corner is $p$ on a backward diagonal, and its lower left corner is a symmetric point $p'$ on a forward diagonal.

Partition the area of $R(\Delta)$ into a top and a bottom part by a horizontal line $\ell_0$ in ratio (Fig.~\ref{fig:case3b-})
$$\left(\weight(\Delta_2)+\frac{2}{3}\delta_2\right)
: \frac{\delta_2}{3}.$$
Partition the top area of $R(\Delta)$ into a left, middle, and right part by two vertical lines, $\ell_1$ and $\ell_2$, in ratio
$$\frac{\delta_2}{3} : \weight(\Delta_2) : \frac{\delta_2}{3}.$$
Let $R_0(\Delta_1)$ be the lower half of the part of $R(\Delta)$ below $\ell_0$, as indicated in Fig.~\ref{fig:case3b}. Let $R_0(\Delta_2)$ be the upper middle part of $R(\Delta)$. Let $R_0(\Delta_3)$ be the right half of the upper right part of $R(\Delta)$. Note that $\area(R_0(\Delta_2))=\area(R(\Delta))\weight(\Delta_2)/\weight(\Delta)$. We choose $Q$ to be the rectangle above $R_0(\Delta_1)$, below $\ell_0$, to the right of $R_0(\Delta_2)$, and to the left of $R_0(\Delta_3$). Observe that placing vertex $v$ at any point in $Q$ yields $R_0(\Delta_i)\subset \Box(\Delta_i)$ for $i=1,2,3$.

We show that $Q$ contains a point on backward diagonal of $B_n$. The width of $Q$ is bounded by
\begin{equation}\label{eq:wq}
w(Q) =\frac{1}{2} w(R(\Delta))\cdot \frac{\delta_2/3}{\weight(\Delta_3)+2\delta_2/3}
\geq \frac{w(R(\Delta))\delta_2}{6\weight(\Delta)}.
\end{equation}
Similarly, the height of $Q$ is bounded by
\begin{equation}\label{eq:hq}
h(Q)
=\frac{1}{2} h(R(\Delta))\cdot \frac{\delta_2/3}{\weight(\Delta)}
= \frac{h(R(\Delta))\delta_2}{6\weight(\Delta)}.
\end{equation}
Using the trivial bound $\delta_2\geq 1$ and \eqref{eq:I33}, we derive $w(Q)> 7/6>1$ and $h(Q)>7/6>1$. However, we have $4\delta_2 n^\alpha> w(R(\Delta))$ or $4\delta_2 n^\alpha >h(R(\Delta))$  in Case~3B. Combining $4\delta_2 n^\alpha >h(R(\Delta))$ with \eqref{eq:wq}, we obtain
$$w(Q)
> \frac{w(R(\Delta))h(R(\Delta))}{24n^\alpha\weight(\Delta)}
=    \frac{\area(R(\Delta))}{24n^\alpha\weight(\Delta)}
\geq \frac{100n\weight(\Delta)}{24n^\alpha\weight(\Delta)}
=\frac{25n^{1-\alpha}}{6}.$$
Similarly, the combination of $4\delta_2 n^\alpha >w(R(\Delta))$ and \eqref{eq:hq} gives
$$h(Q)
> \frac{h(R(\Delta))w(R(\Delta))}{24n^\alpha\weight(\Delta)}
=    \frac{\area(R(\Delta))}{24n^\alpha\weight(\Delta)}
\geq \frac{100n\weight(\Delta)}{24n^\alpha\weight(\Delta)}
=\frac{25n^{1-\alpha}}{6}.$$

Consequently, there is a point $p\in B_n\cap Q$ on a backward diagonal by Lemma~\ref{lem:dense2} when $\alpha\leq 1/2$. Since $R(\Delta)$ satisfies invariant~$I_2$, there is a point $p'\in B_n$ on a forward diagonal with the same $y$-coordinate as $p$, to the left of $R_0(\Delta_2)$. We place vertex $v$ at $p$; define $R(\Delta_1)=\Box(\Delta_1)\cap R(\Delta)$, $R(\Delta_3)=\Box(\Delta_3)\cap R(\Delta)$. Let $R(\Delta_2)$ be a rectangle in $\Box(\Delta_2)$ such that its lower left and lower right corners are $p'$ and $p$, respectively, and its top side is contained in the top side of $R(\Delta)$.

Note that $R_0(\Delta_i)\subset R(\Delta_i)$, for $i=1,2,3$. This establishes \eqref{eq:I3} for $\Delta_2$. Since $\delta_2=\weight(\Delta_1)+\weight(\Delta_3)+1$, this also implies the bound $\area(R(\Delta_i)) \geq \frac{25}{3} n\weight(\Delta_i)$ for $i=1,3$,
which immediately implies \eqref{eq:I4} for $\Delta_1$ and $\Delta_3$.

\smallskip
This concludes the description of the embedding algorithm. Since all invariants are maintained, our algorithm embeds $G$ into $B_n$. The function $\tau$ maps this embedding to a straight-line embedding in $S_n$, completing the proof of Theorem~\ref{thm:main}.

\section{Conclusion\label{sec:con}}

We have presented a set $S_n$ of $O(n^{3/2}\log n)$ points in the plane such that every $n$-vertex planar 3-tree has a straight-line embedding where the vertices are mapped into $S_n$. We do not know what is the minimum size of an $n$-universal point set for planar 3-trees. The point set $S_n$, $n\in \mathbb{N}$, certainly admits some other $n$-vertex planar graphs, as well. It remains to be seen whether $S_n$ is $n$-universal for all $n$-vertex planar graphs.

\paragraph{Acknowledgements.} We are grateful to Vida Dujmovi\'c and David Wood for their encouragement and for repeatedly posing the universal point set problem for 2-trees and planar 3-trees.


\begin{thebibliography}{10}
\itemsep -2pt

\bibitem{ABK+11}
P.~Angelini, G.~{Di Battista}, M.~Kaufmann, T.~Mchedlidze, V.~Roselli, and
  C.~Squarcella, Small point sets for simply-nested planar graphs, in: {\em
  Proc. 19th Symposium on Graph Drawing (GD'11)}, vol.~7034 of {\em LNCS},
  Springer, 2012, pp.~75--85.

\bibitem{BCDE13+}
M.~J.~Bannister, Z.~Cheng, W.~E.~Devanny, and D.~Eppstein,
Superpatterns and universal point sets, arXiv:1308.0403.

\bibitem{Bie11}
T.~Biedl,
Small drawings of outerplanar graphs, series-parallel graphs, and other planar graphs,
\emph{Discrete Computational Geometry} {\bf 45} (2011), 141--160.

\bibitem{BV11}
T.~Biedl and M.~Vatshelle,
The point-set embeddability problem for plane graphs,
in: \emph{Proc. Symposuim on Computational Geometry}, ACM Press, 2011, pp.~41--50.

\bibitem{BV13}
T.~Biedl and L.~E.~Ruiz Vel\'azquez,
Drawing planar 3-trres with given face areas,
\emph{Computational Geometry: Theory and Applications} {\bf 46} (2013), 276--285.

\bibitem{Bos02}
P.~Bose, On embedding an outer-planar graph in a point set,
\emph{Computational Geometry: Theory and Applications} {\bf 23} (3)  (2002), 303--312.

\bibitem{Bra08}
F.-J. Brandenburg, Drawing planar graphs on $\frac{8}{9}n^2$ area,
\emph{Electronic Notes in Discrete Mathematics} {\bf 31} (2008), 37--40.

\bibitem{BMN11}
B.~Bukh, J.~Matou\v{s}ek, and G.~Nivasch,
Lower bounds for weak epsilon-nets and stair-convexity,
\emph{Israel Journal of Mathematics} {\bf 182} (2011), 199--228.

\bibitem{Cab06}
S.~Cabello,
Planar embeddability of the vertices of a graph using a fixed point set is {NP}-hard,
\emph{Journal of Graph Algorithms and Applications} {\bf 10} (2)  (2006), 353--363.

\bibitem{CK97}
M.~Chrobak and G.~Kant,
Convex grid drawings of 3-connected planar graphs,
\emph{Internat. J. Comput. Geom. Appl.} {\bf 7} (1997), 211--223.

\bibitem{CK89}
M.~Chrobak and H.~J.~Karloff,
A lower bound on the size of universal sets for planar graphs,
\emph{SIGACT News} {\bf 20} (4) (1989), 83--86.

\bibitem{CP95}
M.~Chrobak and T.~Payne,
A linear time algorithm for drawing a planar graph on a grid,
\emph{Information Processing Letters} {\bf 54} (1995), 241--246.

\bibitem{FPP90}
H.~de~Fraysseix, J.~Pach, and R.~Pollack, How to draw a planar graph on a grid,
\emph{Combinatorica} {\bf 10} (1)  (1990), 41--51.

\bibitem{BF09}
G.~{Di Battista} and F.~Frati, Small area drawings of outerplanar graphs,
\emph{Algorithmica} {\bf 54} (1)  (2009), 25--53.

\bibitem{DLT84}
D.~Dolev, F.~T. Leighton, and H.~Trickey, Planar embedding of planar graphs,
  in: \emph{Advances in Computing Research}, F.~Preparata, Ed., vol.~2. JAI
  Press Inc., London, 1984.

\bibitem{DEL+11}
V.~Dujmovi\'{c}, W.~Evans, S.~Lazard, W.~Lenhart, G.~Liotta, D.~Rappaport, and
  S.~Wismath, On point-sets that support planar graphs,
  \emph{Computational Geometry: Theory and Applications} {\bf 46} (1) (2013), 29--50.

\bibitem{DM12}
S.~Durocher and D.~Mondal,
On the hardness of point-set embeddability,
in: \emph{Proc. 6th Workshop on Algorithms and Computation (WALCOM'12)},
  LNCS~7157, Springer, 2012, pp.~148--159.

\bibitem{DM13}
S.~Durocher and D.~Mondal,
Plane 3-trees: embeddability and approximation,
in \emph{Proc. 13th Algorithms and Data Structures Symposium (WADS)}
LNCS~8037, Springer, 2013, pp.~291--303.

\bibitem{DMN+11}
S.~Durocher, D.~Mondal, R.~I. Nishat, M.~S. Rahman, and S.~Whitesides,
Embedding plane 3-trees in $\mathbb{R}^2$ and $\mathbb{R}^3$,
in: \emph{Proc. 19th Symposium on Graph Drawing (GD'11)}, LNCS~7034,
  Springer, 2012, pp.~39--51.

\bibitem{ELL+10}
H.~Everett, S.~Lazard, G.~Liotta, and S.~Wismath,
Universal sets of $n$ points  for one-bend drawings of planar graphs with $n$ vertices,
\emph{Discrete and Computational Geometry} {\bf 43} (2) (2010), 272--288.

\bibitem{Far48}
I.~F\'ary, On straight lines representation of plane graphs,
\emph{Acta Scientiarum Mathematicarum (Szeged)} {\bf 11} (1948), 229--233.

\bibitem{HMRS12}
M.~I.~Hossain, D.~Mondal, M.~S.~Rahman, and S.~A.~Salma,
Universal line-sets for drawing planar 3-trees,
in \emph{Proc. 6th Conference on Algorithms and Computation (WALCOM'12)},
LNCS~7157, Springer, 2012, pp 136-147

\bibitem{Fra10}
F.~Frati,
Lower bounds on the area requirements of series-parallel graphs,
\emph{Discrete Mathematics and Theoretical Computer Science} {\bf 12} (5)  (2010),   139--174.

\bibitem{FP07}
F.~Frati and M.~Patrignani,
A note on minimum-area straight-line drawings of planar graphs,
in: {\em Proc. 15th Symposium on Graph Drawing (GD'07)}, LNCS~4875, Springer, 2008, pp.~339--344.

\bibitem{GMPP91}
P.~Gritzmann, B.~Mohar, J.~Pach, and R.~Pollack,
Embedding a planar triangulation with vertices at specified positions,
\emph{American Mathematic Monthly} {\bf 98} (1991), 165--166.

\bibitem{Kur04}
M.~Kurowski,
A 1.235 lower bound on the number of points needed to draw all $n$-vertex planar graphs,
\emph{Information Processing Letters} {\bf 92} (2004), 95--98.

\bibitem{Mon12}
D.~ Mondal, Embedding a planar graph on a given point set,
M.Sci. thesis, University of Manitoba, Winnipeg, MB, 2012.

\bibitem{NMR12}
R.~Nishat, D.~Mondal, and M.~S. Rahman,
Point-set embeddings of plane 3-trees,
\emph{Computational Geometry: Theory and Applications} {\bf 45} (3)  (2012), 88--98.

\bibitem{Sch90}
W.~Schnyder,
Embedding planar graphs in the grid,
in: \emph{Proc. 1st Symposium on Discrete Algorithms}, ACM Press, 1990, pp.~138--147.

\bibitem{ZHN12}
X. Zhou, T. Hikino, T. Nishizeki,
Small grid drawings of planar graphs with balanced partition,
\emph{Journal of Combinatorial Optimization} {\bf 24} (2) (2012), pp.~99--115.

\end{thebibliography}
\end{document}